\documentclass[11pt]{article}
\usepackage[utf8]{inputenc}
\usepackage{cite}
\usepackage{hyperref}
\usepackage{amsmath, amsthm}
\usepackage{nameref}
\usepackage[shortlabels]{enumitem}
\usepackage{todonotes}
\usepackage{graphicx}
\usepackage{subfig}
\usepackage{float}
\usepackage{natbib}

\newcommand{\eps}{\varepsilon}

\newtheorem{theorem}{Theorem}
\newtheorem{lemma}{Lemma}

\theoremstyle{remark}

\newcommand{\ETH}{\texttt{ETH}}
\newcommand{\USDC}{\texttt{USDC}}

\hypersetup{
     colorlinks=true,      		
     linkcolor=blue,        
     citecolor=magenta,     
     filecolor=green,      		
     urlcolor=orange,           	
}

\newcommand\blfootnote[1]{
    \begingroup
    \renewcommand\thefootnote{}\footnote{#1}
    \addtocounter{footnote}{-1}
    \endgroup
}

\title{Fair Interest Rates Are Impossible for Lending Pools: Results from Options Pricing}
\author{Joe Halpern$^1$, Rafael Pass$^{1,2}$, Aditya Saraf$^1$\blfootnote{Joe Halpern and Aditya Saraf were supported in part by NSF grant FMitF-2319186, ARO grant W911NF-22-1-0061, MURI grant W911NF-19-1-0217, and a grant from the Cooperative AI Foundation.}}
\date{$^1$Cornell University\\
$^2$Tel-Aviv University\\
[2ex]
\today}

\begin{document}

    \maketitle
\begin{abstract}
Cryptocurrency \emph{lending pools} are services that allow lenders to pool
together assets in one cryptocurrency and loan it out to borrowers
who provide collateral worth \textit{more} (than the loan) in a
separate cryptocurrency. Borrowers can repay their loans to reclaim
their collateral unless their loan was \textit{liquidated}, which
happens when the value of the collateral dips significantly. Interest rates for these pools are currently set via supply and demand heuristics, which have several downsides, including inefficiency, inflexibility, and being vulnerable to manipulation. Here, we reduce lending pools to options, and then use
ideas from options pricing to search for fair interest rates for lending
pools. In a simplified model where the loans have a fixed duration and can only be repaid at the end of the term, we obtain analytical pricing results. We then consider a more realistic model, where loans can be repaid dynamically and without expiry. Our main theoretical contribution is to show that fair interest
rates \textit{do not exist} in this setting. We then show that impossibility results generalize even to models of lending pools which have no obvious reduction to options. To address these negative results, we introduce
a model of lending pools with fixed fees, and model the ability of
borrowers to top-up their loans to reduce the risk of liquidation. As
a proof of concept, we use simulations to show how our model's
predicted interest rates compare to interest rates in practice. 
\end{abstract}


\section{Introduction}
Decentralized finance (DeFi) involves the use of blockchain
technology, including cryptocurrencies and smart contracts, to provide
decentralized alternatives to traditional financial services like
lending, currency exchanges, or derivatives trading. \emph{Lending pools} are
DeFi services where several users (the lenders) pool together their
cryptocurrency to loan out, charging interest to borrowers.
Traditional loans are often collateralized by non-liquid assets, such as a car, house or business. But DeFi lending pools are completely decentralized and anonymous, and so lending pools require borrowers to provide liquid collateral worth more than the loan
they receive. This collateral is provided in a different
cryptocurrency -- typically, lenders will pool together a more stable
currency, like \USDC{} (USD Coin), and accept a more volatile currency, like
\ETH{} (Ethereum). Because lending pool loans are overcollateralized,
they do not 
serve the typical purpose of traditional loans: increasing one's
liquidity. Instead these lending pool loans can be used to increase
the borrower's exposure to the cryptocurrency deposited as collateral.  

Consider the example of a lending pool that gives out loans of \USDC{}
and holds \ETH{} as collateral. Suppose the current market rate is
$1\ \ETH{} = 100\ \USDC{}$ and that the \textit{overcollateralization
  parameter}, $c$, is $1.5$; this means that depositing $1$ \ETH{} as
collateral would yield a loan of $100/1.5 \approx 67$ \USDC{}. The
borrower can then purchase $0.67\ \ETH$ on the market. If
\ETH\ rises to $110$ \USDC, the borrower can sell their $0.67$
\ETH\ for $73.7$ \USDC\ on the market, and then repay their loan,
spending $67$ \USDC\ to reclaim their $1$ \ETH. Their net worth
in \USDC\ would then be $110 + 73.7 - 67 = 116.7$, an increase of $16.7\%$
from a $10\%$ rise in the value of \ETH, due to their increased
exposure to \ETH. 

Lending pools loans also have a risk of \textit{liquidation}. The
\textit{health factor} of a lending pool loan is the ratio of the
value of the collateral to the value owed. This starts near the
overcollateralization parameter when the loan is initialized, but
could decrease if the collateral dips in value. If the health factor
drops below a liquidation parameter, which we denote $c_0$,
then the loan is subject to liquidation: the borrower's collateral is
sold to third parties at a discount. Typically, $c_0 > 1$, so that
even after a discount, the pool can recoup the value of the
loan.

In practice, lending pools are very popular. Platforms like Aave and
Compound have over \$10.6B USD of total value locked, which is over
$1/8$ of the total value locked in all DeFi protocols at the time of
writing.\footnote{\url{https://www.stelareum.io/en/defi-tvl/lending.html}} But one important feature of lending pools that was omitted in the previous example is that they charge interest rates to borrowers. So, in the example above, the borrower would have to pay slightly more than $67$ \USDC\ to reclaim their $1$ \ETH\ collateral. How should these interest rates be set? In practice, lending pools use supply and demand to set these interest rates. Lenders pool together currency, and increase the interest rate for borrowing that currency as more of it is \textit{utilized} (i.e., lent out). 
This is typically done via piecewise linear functions: the interest rate slowly increases until an inflection point, after which the interest more rapidly increases. This approach, though simple, has many disadvantages. 

One issue is that these utilization functions are not theoretically justified. This leads to several problems. First, without a proper theory, it's difficult to construct the ``best'' utilization function. Big lending pools sometimes offload this problem -- for example, Compound currently
pays \$2 million a year to the risk management firm Gauntlet to help
them set their interest
rate.\footnote{https://www.comp.xyz/t/gauntlet-compound-renewal-2023/4638} Second, even the ``best'' piecewise linear utilization functions are inefficient: \citet{bastankhah2024thinking} shows that such protocols converge very slowly in the presence of slower users, and never reach the true market equilibrium interest. Finally, it's difficult to tell whether or not the interest rates offered by the protocol are favorable as a potential borrower or lender. A better theoretical foundation will address these issues.

Though the problems described above could be addressed by a better theory of utilization functions, utilization functions have other issues in general. First, large suppliers can manipulate the interest rates by withdrawing some of their supply. Doing so would cause the interest rates to increase, and if borrowers are slow to react, could lead to large profits for the suppliers. \citet{subopt2023} formalizes this intuition, describes the optimal lending strategy, and shows that several large suppliers can increase their profits up to 700\% by selectively withdrawing part of their supply. Second, utilization functions lead to more inflexible protocols. While we have mainly focused on the interest rate, the overcollateralization and liquidation parameters are also important. Lending pools might want to learn fair \textit{bundles} of parameters instead of simply fair interest rates. For example, some borrowers may be willing to pay higher interest rates for more capital upfront, while other borrowers may be willing to pay higher interest rates for a lower risk of liquidation. Utilization functions are ill-suited to this finding fair bundles, as a separate market problem would have to be solved for each bundle of parameters.

In this paper, we take a different approach. As the example above shows, lending pool loans increase the borrower's exposure to the collateral. The analog of this in traditional finance is an option. We replicate lending pool loans with \emph{barrier options}, a type of option
that is nullified if the price of the underlying asset falls below a certain barrier. 
Then we can use techniques inspired by option
pricing to find fair interest rates. This approach addresses the limitations of utilization functions: it's easier to find the interest rate via simulations, and borrowers and lenders can evaluate whether the posted interest rate is favorable to them. Large lenders cannot manipulate the interest rate if it's not based on utilization, and it's also possible to find fair ``bundles'' of parameters.
Note that we are not the first to replicate DeFi lending with options; \citet{sardon2021zero} apply option pricing to a new class of \textit{zero-liquidation} loans that they define.\footnote{They also provide a brief outline on how to apply option pricing to standard DeFi lending pools. Our paper takes a different approach to replication than they suggest, and we prove several impossibility results that show that options pricing does not extend straightforwardly from zero-liquidation loans to standard lending pools.}

Our main technical contribution is to show a strong impossibility result: there are no fair lending-pool parameters in a range of models similar to practice. We start with a simplified model, but show that the impossibility result extends to realistic treatments of liquidation as well as time-varying interest rates. We then investigate a model where borrowers pay lenders a fixed fee, which renders fair parameters possible. We provide some simulation results for this setting.

\subsection{Related work}
There has been a flurry of recent work on problems with interest rate models in lending pools. Some of this work has criticized the structure of the utilization functions as piecewise linear functions. \citet{bastankhah2024thinking}, as mentioned earlier, argue that piecewise linear utilization functions converge very slowly and often don't reach the true market equilibrium. \citet{cohen2023economics} provide a model of lending pools based on borrower and lender demand, and show how to optimally update the interest rates to reach a desired utilization. \citet{rivera2023equilibrium} study a model of lending pools where borrowers and lenders arrive randomly and show that the optimal interest rate model must balance between higher efficiency and lower volatility. Importantly, all of these models still set the interest rate as some function of the utilization. Thus, they are vulnerable to strategic utilization attacks and are poorly suited for finding fair bundles of parameters.

Some papers take a more empirical approach. For example, \citet{gudgeon2020defiprotocolsloanablefunds} find that lenders are very unbalanced; sometimes as few as three accounts can control 50\% of the liquidity. They also show that agents do not react quickly to interest rate changes. These facts further support a move away from utilization-based methods which are subject to strategic attacks. \citet{chaudhary2023interest} show that utilization functions on the lending pool Compound tend to have higher slopes for more volatile cryptocurrencies. This supports the idea that different ``fair'' interest rates should be available for different liquidation thresholds, as varying the liquidation threshold is similar to switching to a more volatile currency.

\citet{chiu2022fragility} take the interest rate function as given, and study a model where borrowers can choose which cryptocurrency they wish to deposit as collateral. They show that social welfare improves if the overcollateralization parameters are dynamic, rather than fixed. Our approach can help platforms vary overcollateralization and interest parameters in tandem.

\subsection{Overview}
Our paper proceeds as follows. We first start by providing an overview
of lending pools, barrier options, and the popular Black-Scholes model
for pricing options. In Section 3, we provide analytic formulas for fair interest rates in fixed-term lending
pools. In Section 4, however, we show that a
simple model of dynamic, perpetual lending pools does not admit fair
interest rates. We even extend this impossibility result to other models which do not admit an obvious option replication, such as models with realistic treatments of liquidation. To circumvent the impossibility results, we add a fixed fee to the model, paid by
the borrower to the lender. The remainder of the paper focuses on pricing results in this fixed-fee model. First, in Section 5, we introduce top-ups, which is when a borrower adds additional collateral during their loan term, to avoid liquidation. We
show how top-ups are trivial\footnote{Without an additional discount parameter, the borrower's best strategy is to always top-up by as much money as they have, at the beginning of the loan.} unless we include an additional discount parameter, which is typically
absent in the Black-Scholes model of option pricing. Top-ups are also trivial if we assume that agents continuously monitor their loans, and so we switch to a model of discrete monitoring. Because of all these complications, we must use
simulations, rather than analytic pricing models, to find fair interest rates.  In Section 6, we 
describe our simulation methodology and measure how different model parameters affect the fair interest rates. Finally, we compare our model's interest rates to real interest rates using empirical data from Aave and historical data for market conditions over the last year.

\section{Lending pools, barrier options, and the Black-Scholes model}
We start by formalizing lending pools and barrier options, and then explain the Black-Scholes model, which will be a starting point for the assumptions we make about market conditions and the price movement of the collateral. 

\subsection{Lending pools}
In this section, we provide a model of lending pools without top-ups; we will incorporate top-ups later. We will assume that the lending pool takes \ETH\ as collateral and gives loans in \USDC. Let $1\ \ETH\ = S_0\ \USDC$ at the time of the loan. Let $c > 1$ be the overcollateralization parameter, $1 < c_0 < c$ be the liquidation parameter, and $\alpha > 0$ be the annual interest rate, or APR.

The basic loan action is to deposit $x$ \ETH\ in the lending pool, and
get $S_0x/c$ \USDC\ as a loan. Without loss of generality, we assume
that $x = 1$ because it is trivial to scale up/down for $x \neq
1$.\footnote{For example, if the borrower wants to take a loan of $2$
\ETH, we can consider that a loan of a synethtic currency $\ETH_2$,
which always has double the price of \ETH.} Let $S_t$ be the price of
\ETH\ in \USDC\ at time $t$. In this paper, we consider two ways to
model the loan term. For now, we assume that the borrower can pay back
the loan only at a fixed time $T$ (assuming that the loan was not
liquidated by then). To do so, they pay $e^{\alpha T}S_0/c$ to reclaim
their $1$ \ETH. We later consider the case
where the borrower can pay back their loan at any time in the future.  

The health factor of a loan at time $t$, $H_t$, is defined as the value of the collateral at time $t$ (in \USDC) divided by the value owed. So: $$H_t = \frac{S_t}{e^{\alpha T}S_0/c} = \frac{S_tc}{e^{\alpha T}S_0}\text{.}$$ Liquidation occurs when $H_t < c_0$, which simplifies to $S_t < e^{\alpha T}S_0c_0/c$. 

We make a number of simplifying assumptions in our model of lending
pools. First, we assume that the liquidation mechanism is a costless,
immediate sale for the lenders. This is not true in practice due to
high market friction during liquidation events (see \cite{empirical_defi} for an empirical evaluation of liquidations; \citet{qin2023mitigating} propose a new protocol to reduce the frequency of liquidations). However, our impossibility results apply even to models with a realistic treatment of liquidations. Second, in practice one can often repay partial
loans via lending pools. For example, if you deposit \ETH\ and
\ETH\ goes up in price, then you can withdraw collateral until the
health factor of your loan equals $c$. We consider only full
repayment. This is not a significant limitation, since repaying
a full loan and taking out a smaller loan is a way to simulate partial
repayment. Third, real lending pools, like Aave, often
let users borrow many types of currencies, and store many types of
currencies as collateral. We consider a simple setting, with only
\USDC\ as collateral and \ETH\ that can be borrowed. 

\subsection{Barrier options}
Call options are financial contracts that allow the buyer of the
option to purchase a stock (often called the \textit{underlying}
asset) at some time in the future at a strike/exercise price of
$E$. In this paper, we focus on \textit{barrier} options, a type of
path-dependent\footnote{The value of a standard option depends only on
the price of the underlying asset when the option is exercised; this
is not the case for path-dependent options, which makes pricing such
options more complicated.} option which is either called on or off if
the price of the underlying breaches a barrier, $B$. We are interested
in \textit{down-and-out} call options, where the option is called off
(i.e., canceled) if the price of the underlying falls below the
barrier $B$. The \textit{value} of the option, $v$, is the price that
the option buyer pays to the option seller to purchase the
option. Calculating the fair value of options, under certain
assumptions, is a central question in financial mathematics.  

There are two main types of options, with different exercise
modalities: European options 
and American options. European options have a finite time horizon $t$,
and allow the holder to purchase the stock at exactly $t$. American
options have a potentially infinite time horizon $t$, and allow the
holder to purchase the stock at any time before $t$. If $t = \infty$,
the option is called a \textit{perpetual} American option. Fixed-term
lending pool loans are most similar to European options, but when we
later consider dynamic-term lending pools, we will switch to American
options for replication. 

\subsection{Black-Scholes}
The Black-Scholes model is the foundational (and Nobel-awarded) model
for option-pricing \citep{black-scholes}. It paved the way for the
explosion of options trading, as before their model it was hard for
option buyers and sellers to agree on fair prices. The model makes
assumptions about market conditions, the price movement of the
underlying asset, and utility functions. We highlight the key
assumptions below. 

First, the model assumes that the market is frictionless, which means
that trading can be done without any transaction costs. We adopt this
assumption in our work as well; in reality, transaction costs are an
important dynamic variable for cryptocurrency applications. Second, it
assumes that the market is complete: this means that any ``fair'' trade can
be executed instantly, even for fractional stocks.  Third, it  assumes the
existence of a constant risk-free rate $r > 0$ such that $1$ dollar
now will be worth $e^{rt}$ dollars in $t$ years. In practice, the
10-year US treasury rate is often used as the risk-free rate.  
%
Fourth, the model assumes that agents do not discount the value of
future money beyond the risk-free rate $r$. Of course, agents do
discount the value of future money at a rate of at least $r$, which
essentially means that they calculate their utility for future prices
in today's dollars. We later argue that discounting at a rate higher
than $r$ is important to model top-ups in a nontrivial way.  

Fifth, for the price of the asset, the model assumes an exogenously
given market 
price that follows geometric Brownian motion. This roughly means that
the price varies continuously by multiplying an initial price by a
normally distributed variable. Geometric Brownian motion has a drift
parameter, which describes how the price changes on average over
time. For our purposes, we also assume that the drift of
the stock matches the risk-free rate.
So in expectation, the value of the asset (the price in today's dollars) stays the same over
time. Finally, they make use of the volatility $\sigma_T$ of the asset,
which is the standard deviation of the asset's logarithmic
returns over time period $T$.\footnote{For an asset with price path
following geometric Brownian motion, $\sigma_T$ can be calculated from
the yearly returns $\sigma_{\text{annual}}$ via $\sigma_T =
\sigma_{\text{annual}} \sqrt{T}$.} 

\section{Analytic pricing for fixed-term lending pools}
In this section, we first replicate a lending pool loan using an
option with certain parameters. We use existing option-pricing results to find the fair price of that option, thereby obtaining an analytical function for the fair interest rates in this model.

\subsection{Replication}
\textit{Financial replication} of a given asset is when the cash flow of the asset is replicated by other financial instruments. Here, we will replicate the cash flow of a lending pool loan via a barrier option. 
\begin{theorem}
    \label{european_replication}
    A European option down-and-out call option with exercise price $E = e^{\alpha T}S_0/c$, value $v = (1-1/c)S_0$ and barrier $B = e^{\alpha T}S_0c_0/c$ can be used to replicate both the borrower and lender's utility in a fixed-term lending pool.
\end{theorem}
\begin{proof}
    First, note that the borrower in a fixed-term lending pool chooses
    to repay at time $T$ if $e^{-rT} S_T > e^{-rT} e^{\alpha
      T}S_0/c$. If the loan was not liquidated, then it must be the
    case that $S_T > e^{\alpha T}S_0 c_0/c$ (the liquidation
    threshold). Thus, if the loan survives until $T$ (without
    liquidation), the borrower will always choose to repay.  
This means that the utility (in \USDC) of the loan for the borrower depends on
 which of the following two cases occurs: (1) the borrower repay the
 loan at $T$, or (2) the 
loan is liquidated, and so the borrower simply keeps the initial loan
amount.  Hence, the borrower's utility is
$$u = \begin{cases}
  S_0/c & \text{if the loan is liquidated ($S_t < \frac{e^{\alpha T}S_0c_0}{c}$), } \\
  S_0/c + e^{-rT}(S_T - e^{\alpha T} S_0 /c) &
        \text{if the loan repaid at time $T$.} 
\end{cases}$$

To replicate this, consider a European down-and-out call option with
parameters $E = e^{\alpha T}S_0/c$,  $v = (1 - 1/c) S_0$,  $B =
\frac{e^{\alpha T}S_0c_0}{c}$.  First, the option buyer will sell his
initial $1$ \ETH\ for $S_0$ \USDC, and then buy the option above. Then
they will simply exercise the option at the end of their term if the
option survives until then. If the barrier is breached (at $S_t <
\frac{e^{\alpha T}S_0c_0}{c}$), the buyer would have his initial 
utility of $S_0 - (1-1/c)S_0  = S_0/c$. If the option is exercised, the buyer
would have a utility of $S_0/c + e^{-rT}(S_T - e^{\alpha T}S_0/c)$. This
matches the utility of the lending pool borrower. 

We now consider the perspective of the lender. The lender's utility
depends on the same cases as the buyer: (1) the loan is repaid at $T$,
(2) the loan is liquidated (in which case the lender simply sells the
collateral at the liquidation-threshold price). Thus, their utility
is: 
$$u = \begin{cases}
  e^{-rT}e^{\alpha T} S_0 /c - S_0/c  & \text{if the loan is repaid,}\\
    e^{-rt}\frac{e^{\alpha T}S_0c_0}{c} - S_0/c & \text{if the loan is
    liquidated at time $t$.} 
\end{cases}$$

We replicate the lender's utility using the same option with the same
parameters. The option seller must convert the price of the option
(paid in \USDC) to \ETH\ immediately and hold that. Then they
purchase $1/c$ \ETH\ on the spot market immediately, in order to hold
$1$ \ETH. If the option is liquidated, they immediately sell their $1$
\ETH\ on the market.
%
If the option is exercised, they receive a payment of $e^{\alpha
T}S_0/c$, so their total utility is $e^{-rT}e^{\alpha T}S_0/c -
  S_0/c$. If the option breaches the barrier at time $t$, they
immediately sell for $e^{\alpha T}S_0c_0/c$, so their final
  utility is $e^{-rt}e^{\alpha T}S_0c_0/c - S_0/c$. This matches the
utility of the lender.  
\end{proof}

Taking a step back, we've replicated the utility of a lending pool
loan using an option with parameters $v = (1 - 1/c) S_0 $, $E = e^{\alpha T}S_0 /c$, and $B = \frac{e^{\alpha T}S_0c_0}{c}$. How is this
helpful? There has been extensive research on how to fairly price
options, where ``pricing'' means determining the fair value $v$ given
the other option parameters and some market parameters. For European
barrier options, there even exist analytic pricing formulas for the
fair value $v$, so we would simply set that equal to $(1 - 1/c)
S_0$ to solve for the parameter we are interested in (e.g., $\alpha$).  


\subsection{Pricing fixed-term lending pools via European barrier options}
The pricing formulas for European barrier options depend on whether the barrier price is higher or lower than the exercise price. In our example, it's easy to see that the barrier price is always higher than the exercise price, since $c_0 > 1$. 

\begin{theorem}[Adapted from \cite{Haug_2007}]
    \label{european_value}
    Suppose that the risk-free rate is $r$ and the volatility of the asset is $\sigma$. A European down-and-out call option of duration $T$ with parameters $E = e^{\alpha T}S_0/c$ and $B = e^{\alpha T}S_0c_0/c$ has value:
    \begin{align*}
        v = S_0e^{-rT}\left(N(\eta_1) - \frac{e^{\alpha T}}{c}N(\eta_1-\sigma_T) - \frac{e^{\alpha T}c_0}{c}N(\eta_2) + \frac{N(\eta_2 - \sigma_T)}{c_0}\right),
    \end{align*}
    where:
    \begin{align*}
        N(\cdot)&\text{ is the CDF of a standard normal,} \\
        \eta_1 &= \frac{\log\left(\frac{c}{e^{\alpha T}c_0}\right)}{\sigma_T} + \frac{\sigma_T}{2}, \text{ and } \\
        \eta_2 &= \frac{\log\left(\frac{e^{\alpha T}c_0}{c}\right)}{\sigma_T} + \frac{\sigma_T}{2}.
    \end{align*}
\end{theorem}
Now we can use this theorem to price the lending-pool parameters.

\begin{theorem}
\label{fixed-term-eq}
    The fair values for the lending-pool parameters are those which satisfy the following equation:
    $$1-1/c = e^{-rT}\left(N(\eta_1) - \frac{e^{\alpha T}}{c}N(\eta_1-\sigma_T) - \frac{e^{\alpha T}c_0}{c}N(\eta_2) + \frac{N(\eta_2 - \sigma_T)}{c_0}\right).$$
\end{theorem}
\begin{proof}
    From \autoref{european_replication}, we know that the value of the option that replicates the borrower and lender's utilities is $S_0(1 - 1/c)$. From \autoref{european_value}, we know that the fair value of this option is $S_0e^{-rT}\left(N(\eta_1) - \frac{e^{\alpha T}}{c}N(\eta_1-\sigma_T) - \frac{e^{\alpha T}c_0}{c}N(\eta_2) + \frac{N(\eta_2 - \sigma_T)}{c_0}\right)$. By setting these values equal and cancelling the $S_0$ term on both sides, we get the equation for fair lending-pool parameters.
\end{proof}
Note that there are three lending-pool parameters, but only one equation, so there is no unique solution for all three parameters. In this paper, we focus on the interest rate ($\alpha$), and thus assume that the liquidation and overcollateralization parameters are fixed. But in practice, all three of these parameters could be adjusted. For instance, a lending pool might offer two sets of fair loans: a lower-interest loan, with higher overcollateralization and liquidation parameter, and a high-interest loan, with lower overcollateralization and liquidation parameters. 
\section{Dynamic, perpetual lending pools}
We now move to a more realistic model of dynamic, perpetual loans. Loans
can now be repaid any time before liquidation. Also,
interest continuously accrues, so the health factor of the loan
decreases over time in expectation (if the interest rate is above the
risk-free rate). We will replicate this using perpetual American
barrier options, rather than European barrier options.  

We first prove an impossibility result in this model: the
borrowers are always (weakly) favored, and the only parameters that are fair are trivial parameters where the borrower's optimal strategy is to immediately repay the loan. We expand the impossibility result to other models, discuss
the ramifications of these results in practice, and add fixed trading fees to
this model as a solution.  

\subsection{Dynamic, perpetual lending pools have no  fair parameters}
We first provide a model of these lending pools. 
As before, let $\alpha$ now represent the yearly interest rate, and
assume that interest is continuously compounded. Now, the borrower can
pay back the loan at any point before liquidation. If they repay the
loan $t$ years from when they received it, they must pay $e^{\alpha
  t}S_0/c$ to reclaim their $1$ \ETH. The health factor now
  becomes $H_t = \frac{S_t c}{e^{\alpha t}S_0}$. Liquidation occurs when
$H_t < c_0$, which simplifies to $S_t < e^{\alpha t}S_0 c_0/c$. Note
that the liquidation barrier now increases over time, due to the
accumulating interest.  

Switching to dynamic loans means that we must now explicitly consider
the \textit{monitor frequency} of the borrowers: how often do they
monitor their loans to decide when they should repay? For now, we
continue in the spirit of Black-Scholes, and assume that borrowers
continuously monitor their loans. As a consequence, when their loan
reaches the liquidation threshold, they will always have the option to
repay their loan rather than letting it liquidate. 

\begin{theorem}
    \label{american-replication}
If borrowers continuously monitor their loans then a
    perpetual American down-and-out call option with exercise price
  $E_t = e^{\alpha t} S_0 /c$, value $v = (1 - 1/c) S_0$, and
    barrier $B_t = \frac{e^{\alpha t} S_0c_0}{c}$ can be used to
    replicate both the borrower and lender's utility for a dynamic,
    perpetual lending pool. 
\end{theorem}
\begin{proof}
    First, note that borrowers will always choose to repay their loans at the liquidation threshold. If they repay at the liquidation threshold of $S_t = e^{\alpha t}S_0 c_0/c$, their utility is $S_0/c + e^{(\alpha - r) t}S_0 c_0/c - e^{(\alpha - r) t}S_0/c$; if they let the loan liquidate their utility is just $S_0/c$. Since $c_0 > 1$, they will always prefer to repay at the liquidation threshold. Thus, their overall utility is:
    $$u = \begin{cases}
        \frac{S_0}{c} (1 + e^{(\alpha-r) t}(c_0 - 1)) & \text{if  $S_t
                    = \frac{e^{\alpha t}S_0c_0}{c}$},\\ 
        e^{-rt}(S_t - e^{\alpha t} S_0 /c) + S_0/c & \text{if the loan
          repaid at time $t$.} 
    \end{cases}$$
    
    To replicate this, we will use a perpetual American down-and-out call option. This is an option with no expiration date, that can be exercised anytime before the barrier is breached to buy the underlying asset. Further, both the exercise price and barrier for the option will increase with time.
Suppose that this option has the following parameters:
    \begin{itemize}
\item exercise price $E = e^{\alpha t} S_0 /c$,
\item value of option $v = (1 - 1/c) S_0$,
\item barrier $B = \frac{e^{\alpha t} S_0c_0}{c}$.
      
    \end{itemize}
    The replication is identical to
    \autoref{european_replication}. The option buyer will sell his
    initial $1$ \ETH\ for $S_0 $ \USDC, and then buy the
option  above. They will then simply exercise the option exactly when
    they would repay the loan. If the barrier is hit (but not
    breached), the buyer will always choose to exercise the option
    rather than let it be canceled. In this case, their utility would be
    $S_0/c + e^{-rt}(\frac{e^{\alpha t}S_0 c_0}{c} - e^{\alpha
      t}S_0/c) = \frac{S_0}{c} (1 + e^{(\alpha-r) t}(c_0 - 1))$. If
   the option is exercised at time $t$, the buyer would have a utility
   of $S_0/c +  e^{-rt}(S_t - e^{\alpha t}S_0/c)$. This matches the
    utility of the lending-pool borrower. 
    
    The lender's utility is even simpler to calculate than the borrower's
    utility. The borrower will always repay the loan (either at the
    liquidation threshold or earlier), and thus the lender will simply
    receive $e^{\alpha t} S_0/c$ when that happens. Including the
    amount loaned to the borrower, the lender's total utility is thus $u =
    (e^{(\alpha-r) t} - 1)S_0/c$, where $t$ is the time the borrower
    repays the loan.  

The procedure for replicating the lender's utility is also identical
to that given in the proof of 
    \autoref{european_replication}, so we do not step through it in
    detail. 
\end{proof}

We've reduced the challenge of finding fair lending-pool
parameters to that of determining a fair value for these American barrier
options. Unfortunately, we now show that there are no fair parameters
for the American barrier option used to replicate the lending pool. 

\begin{lemma}
      A perpetual American down-and-out call option over an underlying
    asset of value $S_0$, with (possibly varying) exercise price $E_t$
    and barrier $B_t$ has value at least $S_0 - E_0$. The value is
    equal to $S_0 - E_0$ only in the trivial case where the optimal
    strategy for the buyer is to immediately exercise the option. 
    \label{lm:american_bound}
\end{lemma}
\begin{proof}
    Since this option is American, the buyer can always choose to exercise it immediately. If they do, they would receive utility $S_0 - E_0$. The actual utility of the option must be at least this value. 
    %
    If the buyer had any positive expected utility for holding the option after $t= 0$, then the value would be strictly greater than $S_0 - E_0$. This proves the second statement in the lemma.
\end{proof}
\begin{theorem}
    In a dynamic, perpetual lending pool with no fixed fees, there are no fair values for $\alpha$, $c$, and $c_0$ such that the borrower wants to hold the loan for any time.
    \label{thm:perpetual_negative}
\end{theorem}
\begin{proof}
    By applying \autoref{lm:american_bound} to the option in the
    replication in \autoref{american-replication}, we get that the
    value of the option is at least $S_0 - E_0 = S_0 - e^{\alpha \cdot
      0}S_0/c = S_0(1-1/c)$. If $\alpha$, $c$, and $c_0$ are set such
    that the borrower does not want to immediately exercise the
    option, then the value of the option must be greater than $S_0 -
    E_0$. In this case, there are no fair parameters, since the
    value exceeds $S_0(1 - 1/c)$, which is the value required for
    replication.  
\end{proof}

\subsection{Expanding the impossibility result}
We now expand the impossibility result to different models of lending pools. The lemmas in this section are technically simple but conceptually interesting because they can apply to models of lending pools that do not admit an option replication. For instance, suppose we want to more accurately model the loss that lenders experience due to liquidation. To do this, we would need to model a third party, the liquidators, which would render option replication more difficult (as there are only two parties involved). But we can apply \autoref{lm:lenders-worse} to such a model, and prove that it would suffer from the same impossibility result, without an explicit option replication.

In the context of finding fair parameters, the main problem with our original
model is that it overly favors the borrowers. We formalize this intuition with the following two lemmas. To do so, let (IMP) refer to following:
\begin{table}[H]
    \centering
    \begin{tabular}{p{1cm} p{10cm}}
         (IMP) &  For any set of lending-pool parameters $P = \{\alpha, c, c_0\}$, 
         either the borrower's optimal strategy is to immediately repay the loan, or the borrower is favored over the lender. 
    \end{tabular}
\end{table}

\begin{lemma}\label{lm:lenders-worse}
    Let $M$ denote a model of dynamic, perpetual lending pools which satisfies (IMP). Let $M'$ denote a model of dynamic, perpetual lending pools where lenders always receive at most as much utility as in $M$, but the utility of borrowers is unchanged. Then $M'$ also satisfies (IMP).
\end{lemma}
\begin{proof}
    Consider a set of lending-pool parameters: $P = \{\alpha, c, c_0\}$. We split into two cases, based on the optimal response of the borrowers to $P$ in $M$.

    For the first case, suppose that the optimal response is to immediately repay the loan. Then, in $M'$, the optimal response for the borrowers will also be to immediately repay the loan, since their utility is the same in $M'$ and $M$.

    For the second case, suppose that the optimal response is for the borrowers to hold the loan for some time. By (IMP), this means that the borrowers are favored over the lenders. But since the utility of the lender in $M'$ is at most as high as in $M$, the borrower is still favored in the new model, and so the parameters are unfair.

    Thus, in $M'$, either the borrower immediately repays the loan, or the borrower is favored. So, $M'$ satisfies (IMP) as well. 
\end{proof}
As we mentioned earlier, this lemma is powerful because it applies even to models which do not admit an obvious option replication. Consider a model where lenders lose some amount of money to the liquidator via an auction (or by offering a fixed discount to the liquidators). In our model, we've assumed that liquidations are costless to the lender. Thus, in the more realistic treatment, lenders receive at most as much utility as in our model. The lemma above shows that there are no fair parameters in such a model. 

For another example, note that in practice lenders do not get the full interest that borrowers pay; instead, the lending platform takes a small fee. By \autoref{lm:lenders-worse}, the impossibility results applies to this setting as well. 

We can prove a similar result for borrowers. 
\begin{lemma}
    Let $M$ denote a model of dynamic, perpetual lending pools which satisfies (IMP). Let $M'$ denote a model of dynamic, perpetual lending pools where borrowers always receive at least as much utility as in $M$, but the utility of lenders is unchanged. Then, $M'$ also satisfies (IMP).
\end{lemma}
\begin{proof}
    Consider a set of lending-pool parameters: $P = \{\alpha, c, c_0\}$. We split into two cases, based on the optimal response of the borrowers to $P$ in $M$.

    For the first case, suppose that the optimal response is to immediately repay the loan. If the optimal strategy of the borrower in $M'$ is to immediately repay, we're done. Otherwise suppose that in $M'$, the borrower would rather hold the loan for longer. This means that the borrower gets more utility, since they would get the same utility for immediately repaying in $M$ as in $M'$, by assumption. Thus, the borrower is favored.

    For the second case, suppose that the optimal response is for the borrowers to hold the loan for some time. By (IMP), this means that the borrowers are favored over the lenders. Since the borrower receives at least as much utility in $M'$, the borrower is still favored over the lender. 

    Thus, in $M'$, either the borrower immediately repays the loan, or the borrower is favored. So, $M'$ satisfies (IMP) as well. 
\end{proof}
For an application of this, consider adding the option for borrowers to ``top-up'' their loan, by adding additional collateral part-way through the loan term to lower the risk of liquidation. Such a model would have borrowers receive at least as much utility, as they can always choose to not top-up their loan. This lemma shows that such a model suffers from the same impossibility result.

These lemmas also compose nicely, and can be applied to more complicated models. For example, as we mentioned earlier, interest rates are not static for lending pools in practice; they vary over time due to utilization. Is there some fair ``dynamic'' interest function? Unfortunately not. Let $\alpha_1$ denote the highest interest in an arbitrary execution of the interest function, and let $\alpha_0$ denote the lowest interest. Consider the model where borrowers pay $\alpha_1$ and lenders receive $\alpha_0$. By the same logic as in the ``platform-fee'' model discussed earlier, this model satisfies (IMP). Now, note that the dynamic interest model is (possibly weakly) better for borrowers and worse for lenders. Thus, by composing \autoref{lm:lenders-worse} and \autoref{lm:borrowers-better}, we get that there is no dynamic interest function that is fair for both parties. 

Unfortunately, not \textit{every} way of penalizing borrowers gets around the impossibility result. In particular, consider the following plausible model changes that are likely true in practice:
\begin{enumerate}
    \item Borrowers do not monitor their loans continuously, but instead discretely. This \textit{does} penalize borrowers, who would be exposed to liquidation risk when discretely monitoring their loan. However, this does not fix the problem with our model. It's easy to create an option replication for this model that is nearly identical to the replication in \autoref{american-replication}. Further, even in the new model, borrowers extract the same amount of utility (exactly $S_0(1-1/c)$) from \textit{immediately} repaying the option. Thus, the same impossibility result goes through. Later, when we consider a model that does circumvent the impossibility result, we do add discrete monitoring for borrowers. 
    \item Borrowers might discount future utility beyond the risk-free rate. Later, we add a discount parameter for borrowers, but that cannot fix this problem for the same reason as discrete monitoring.
\end{enumerate}

\subsection{Fair parameters vs. utility-maximizing parameters}
Despite our impossibility results, lending pools are very popular in practice. In this section, we contextualize our impossibility results, and consider reasons why lenders may not prioritize fair parameters.

The intuition behind \autoref{thm:perpetual_negative} is that the
borrower can extract too much value from the lending pool for a fair
replication. Suppose that a lending pool is trying to determine the best
value of $\alpha$ by raising $\alpha$ from $0$ until the borrower no
longer wishes to participate (and suppose that $c > c_0 > 1$ is
fixed). The borrower may wish to participate and hold the loan for
small values of $\alpha$. As the interest rate rises, they would want
to hold the loan for less time. At some point, say $\alpha^*$, the
borrower's best strategy is to simply immediately repay the loan,
which nets $0$ utility. Since the borrower is not holding the loan for
any time, he pays no interest. Thus, the interest rate can be
arbitrarily high, but the pool is fair since neither party is
benefiting.  

To understand this better, it's worth analyzing our replication(s)
further. Our replication schemes replicate the utilities of both the
borrower and lender in the lending pool, by considering the buyer and seller
of an option. In order for the replication to work, the option seller
must convert the price of the option into \ETH, and hold that. Thus,
the option here is a \textit{covered} option, which is an option where
the option seller owns the underlying stock, and thus the buyer does
not have to worry about the seller being unable to honor the
option. However, covered options and lending pools are often framed
differently, despite having the same underlying utilities. Option
sellers might view giving their stock to the buyer as a loss if $S_t >
E_t$, as they could sell the stock for more money on the market if
they weren't restricted by the option. On the other hand, since the
collateral is owned by the buyers, lenders typically view lending
pools as more akin to savings accounts, since their utility does not
depend on $S_t$ at all. That is, they may compare their utility from
the lending pool to a baseline of simply holding their \USDC\ and
obtaining the risk-free rate, rather than buying \ETH\ with their
\USDC\ on the spot market. 

\begin{figure}
  \centering
  \subfloat{\includegraphics[width=0.48\textwidth]{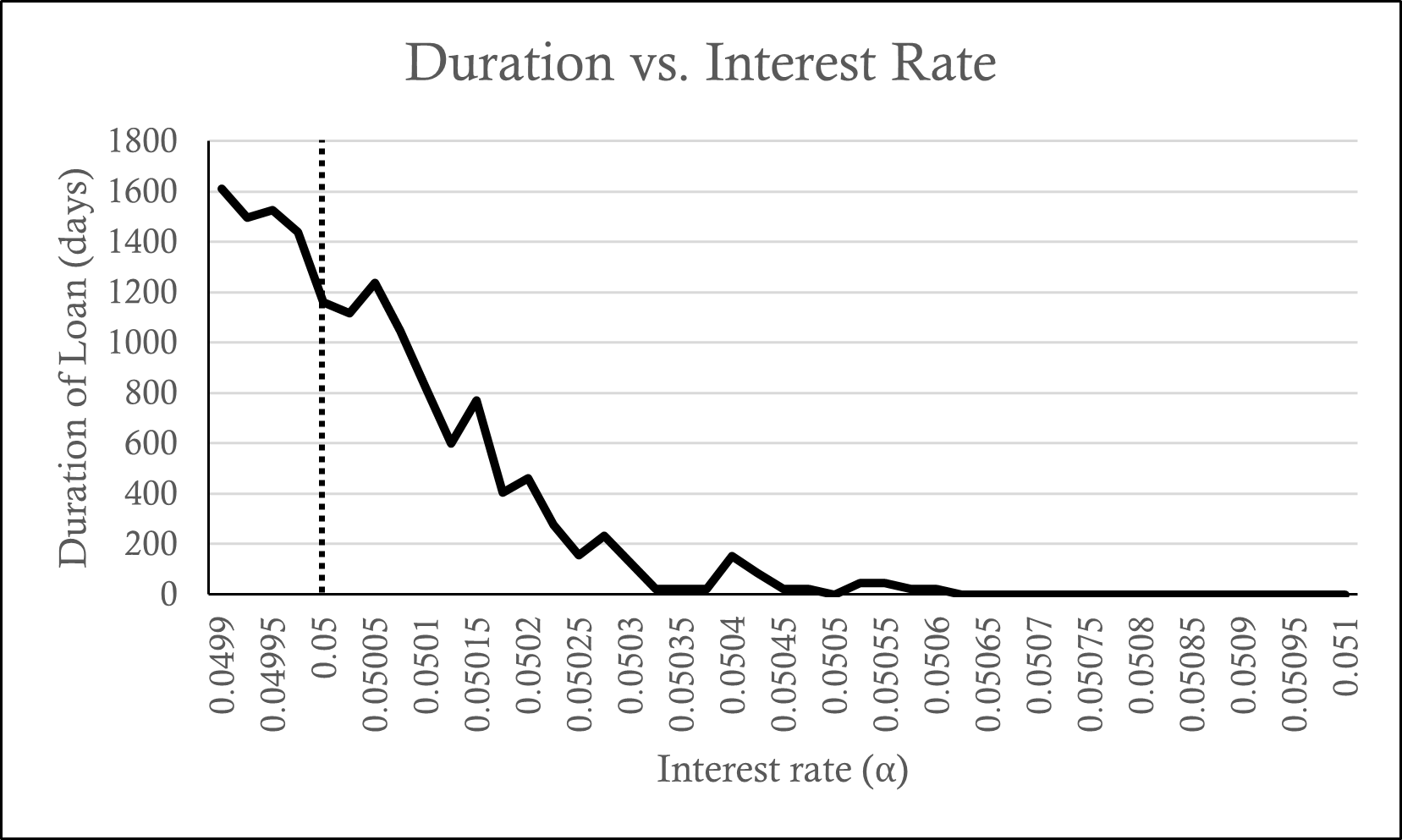}}
  \hfill
  \subfloat{\includegraphics[width=0.48\textwidth]{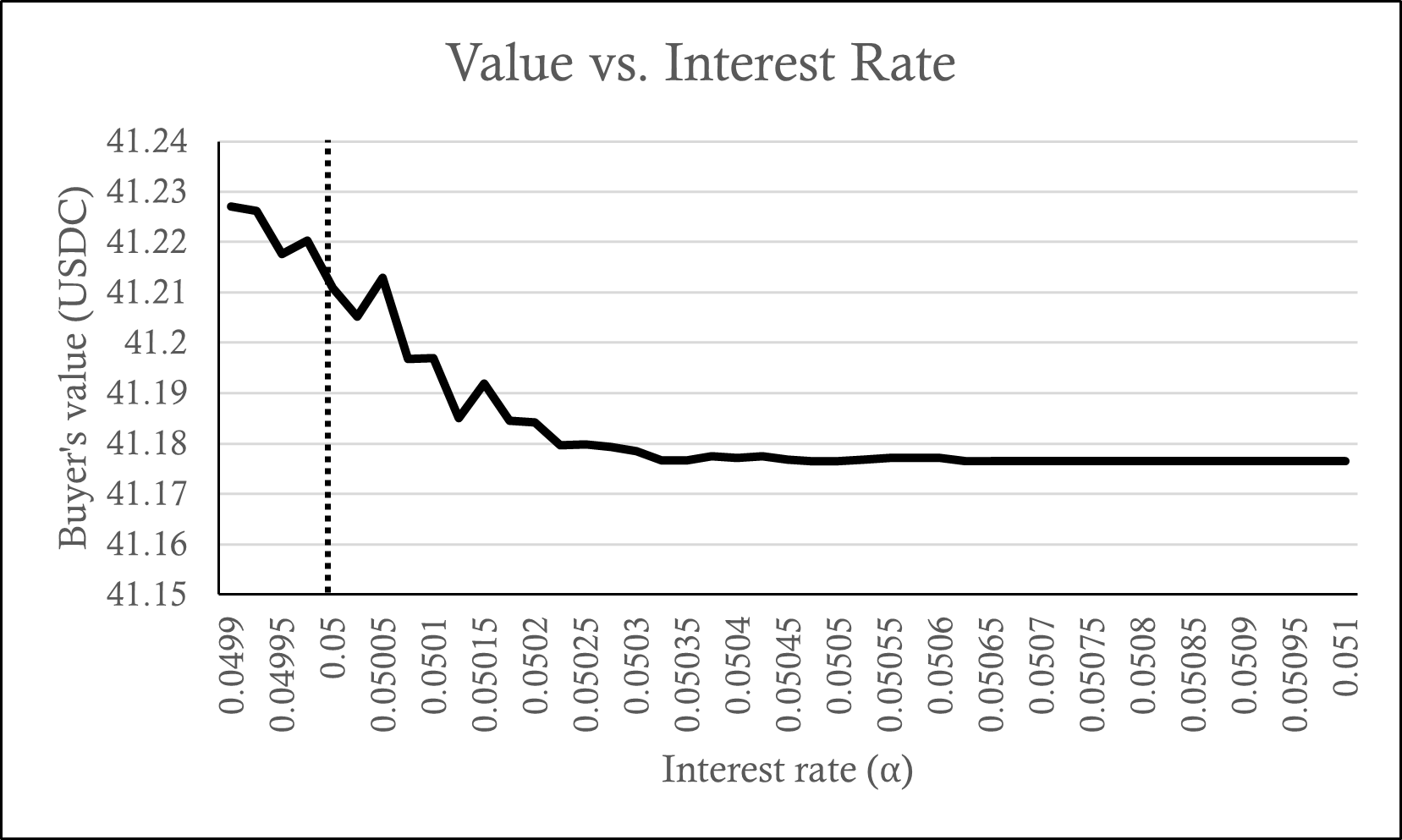}}
  \caption{These graphs show how the duration and buyer's expected value decrease as the interest rate increases. They were generated from simulations; see Section 6 for more details on the simulations. Here, we simulated with $S_0 = 100, c = 1.7, c_0 = 1.2,$ and $r = 0.05$. The dotted line at $\alpha = 0.05$ indicates the smallest $\alpha$ which gives the lender an expected utility higher than the risk-free rate. Note that the value quickly converges to $S_0(1-1/c) \approx 41.17$, as the buyer's optimal strategy converges to immediately exercising the loan. The only fair interest rates are those which yield a buyer value of $S_0(1-1/c)$.}
  \label{fig:duration}
\end{figure}

Recall from the proof of \autoref{american-replication} that the
lender's utility in a dynamic, perpetual lending pool is simply
$(e^{(\alpha - r)t} -1)S_0/c$, where $t$ is the time when the borrower
repays the loan. Clearly, if $\alpha < r$, the lender incurs negative
utility, and they would have preferred to not participate in the
lending pool. What if $\alpha > r$? The lender does receive positive
utility, and thus prefers participation in the lending pool to
non-participation (i.e., just getting the risk-free rate). And in
fact, \autoref{fig:duration} shows a concrete example of this, where there \textit{are} values of $\alpha >
r$ where both the lenders and borrower will have positive expected
utility. This can be seen in the graph above: since the value of the
option is greater than $S_0(1-1/c)$, we know that by the replication
that the borrower has positive expected utility. And we just argued
that when $\alpha > r$, the lender receives positive utility.  

Yet by \autoref{thm:perpetual_negative}, these parameters are
unfair. If we think in terms of the option seller, this means that
there must be some financial portfolio (i.e., some set of trades)
available to the option seller that has greater utility than selling
the option; this is simply what it means for the fair value to be
greater than the amount paid to the seller. Similar remarks apply to
the lender, so the lender receives positive utility as compared to
simply holding their \USDC, but receives negative utility as compared
to the optimal financial portfolio. 

This may seem like a trivial framing issue. But it helps explain why
dynamic, perpetual lending pools exist in practice. If lenders want to
treat lending pools like a savings account, they may be unwilling to
engage in financial portfolios that involve selling their \USDC, and
holding \ETH\ instead. This could be a form of risk aversion, or they
could be sensitive to transaction fees. So, a distinct but well
motivated question one might ask is what are the
\textit{revenue-maximizing} parameters for the lender in the dynamic,
perpetual lending pool? Answering this question would require more
modeling -- if there was a limited supply of borrowers, then the
lending pool would want to balance high interest rate with high loan
duration; if there were was an infinite of borrowers, then the lending
pool would just want to maximize interest rates subject to positive
loan duration. 
In this paper, we focus on the \textit{fair} parameters; parameters
for the lending pools that guarantee that there is no portfolio that
always outperforms the lending pool for either the borrower or the
seller. But we wanted to highlight the question of finding the
revenue-maximizing parameters, both to make the impact of our theorem
more clear, and as an interesting question for future study. 

\subsection{Circumventing the impossibility result}
As we saw earlier, we need to focus on changes to the model that would penalize borrowers or favor lenders. What factors do this in practice? 
\begin{enumerate}
    \item Borrowers may struggle to determine the optimal exercise threshold. In particular, there are no closed-form solutions for the optimal exercise threshold, even given the Black-Scholes model \citep{gao2000valuation}. 
    \item Borrowers might have inaccurate beliefs about the price movement, which lead to suboptimal behavior.
\end{enumerate}


These factors are reasonable behavioral explanations, but
would require more complex modeling. We want to penalize borrowers and/or favor lenders, but do so in a way that is compatible with option replication. Thus, we focus on
another factor: transaction costs. In practice, borrowers are
interacting with the lending pool much more often than lenders, who
often just deposit money and let it accrue interest over time. Thus,
the burden of transaction costs falls on borrowers. Modeling
transaction costs directly in our model is tricky, since it involves
transactions to a third party (a miner or validator) and thus
replication for borrowers and lenders would no longer be
symmetrical. We instead add a fixed fee associated with
repayment. 

\subsection{Replication for fixed fees}
The only change to the previous model is the addition of a fixed
fee $\beta$ that borrowers pay to lenders to reclaim their
collateral. Thus, if they repay the loan $t$ years from when they
received it, they must pay $e^{\alpha t}S_0/c + \beta$ to reclaim
their $1$ \ETH. The health factor is now $H_t =
\frac{S_t}{e^{\alpha t}S_0/c + \beta}$, so liquidation occurs when
$S_t < (e^{\alpha t}S_0/c + \beta)c_0$. 

\begin{theorem}
If borrowers continuously monitor their loans, then a
  perpetual American down-and-out call option with exercise price
     $E_t = e^{\alpha t} S_0 /c + \beta$, value $v = (1 - 1/c) S_0$,
     and barrier $B_t = (e^{\alpha t}S_0/c + \beta)c_0$ can be used to
     replicate both the borrower and lender's utility for a dynamic,
     perpetual lending pool with a fixed fee. 
     \label{fixed_fee_repl}
\end{theorem}
\begin{proof}
    As before, the borrower will always choose to repay their loans at the liquidation threshold. To see this, note that repayment yields a utility of $S_0/c + e^{-rt}(B_t - E_t)$, where $t$ is when the loan hits the liquidation threshold. Liquidation yields a utility of $S_0/c$, so repayment is preferred when $e^{-rt}(B_t - E_t) > 0$, or when $B_t > E_t$. As before, $B_t = c_0 E_t$, and since $c_0 > 1$, repayment is always preferred. 
    
    The remainder of the proof proceeds exactly as in \autoref{american-replication}. The one difference is that the lender's utility is now $e^{-rt}(e^{\alpha t}S_0/c + \beta) - S_0/c$, assuming the borrower repays at time $t$. Similar to before, if $\alpha > r$, then this value is strictly positive, and lenders prefer participation to collecting the risk-free rate. Unlike before, if $\alpha \le r$, lenders may still want to participate, as their utility will be positive if $\beta > S_0/c(1 - e^{(\alpha - r)t})$.
\end{proof}
Let $v'$ be the fair price of the option with the parameters
given in the theorem. Due to the replication above, we can find the
fair interest rate $\alpha$ given some $c$ and $c_0$ by setting $v' =
S_0(1 - 1/c)$.
Applying \autoref{lm:american_bound}, we get that $v' \ge S_0 -
E_0 = S_0(1 - 1/c) - \beta$. So, as $\beta$ increases, the parameters
must be such that the borrower can extract more utility from holding
the option.  

Unfortunately, since this option is a perpetual, American option with dynamic barrier and exercise price, there is no known analytic formula for $v'$. Instead, we will use simulations to estimate $v'$. Before we do that, we make our model more realistic by adding top-ups. 




\section{Top-ups: discounting and discretization}
In real lending pools, agents have the option to top-up their loan by adding more collateral. This increases the health factor of their loan, helping them avoid liquidation. When they repay their loan, they can reclaim their original collateral as well as any new collateral. 

We model top-ups by giving borrowers the power to add an arbitrary amount of collateral to their loan at any time. If the borrower adds $\eps$ collateral at time $t$, the health factor rises to $H_t = \frac{S_t(1 + \eps)}{e^{\alpha t}S_0/c + \beta}$, and so the new liquidation threshold is $S_t < (e^{\alpha t}S_0/c + \beta)\frac{c_0}{1+\eps}$. Like in the previous section, we assume here that agents can continuously monitor the market to decide when to top-up. However, our model as stated so far trivializes top-ups. We explain why in the next two subsections, and change the model to involve discounting and discrete monitoring.

\subsection{Discounting}
The first issue with the model is that agents are overly incentivized
to top-up. To see a simple example of this, consider an agent with a
very large budget. Such an agent, when taking out a loan, could simply
perform a large top-up right at the beginning to lower the liquidation
threshold significantly. Since the price of the collateral is modeled
by geometric Brownian motion (with drift matching the risk-free rate),
the value of the collateral, in expectation, stays the same over
time.\footnote{The \textit{price} of the collateral drifts up over
time, but the agents calculate the value of the collateral in today's
dollars when computing their utility.} Further, with continuous
monitoring, there is no risk that the collateral gets
liquidated. Thus, the top-up in itself is an action with zero net utility.
But it lowers the liquidation threshold greatly, so it
benefits the agent overall. 

There are several ways to fix this issue. One is to give the agent a
fixed budget for topping-up, and investigate how to best spend this
budget. This approach has some potential issues. First, it's not clear
how to model the opportunity cost of spending money on top-ups in the
fixed budget model. Second, in practice, it's unlikely that agents have
a fixed budget; agents can usually take out loans to stretch
their budget at some additional cost.  

We instead introduce a discount parameter $\delta$ into the
model. The value of $x$ dollars two years from now is $e^{-(r +
  \delta)}x$ today, rather than just $e^{-r}x$ (where $r$ is the
risk-free rate, as before). This parameter has several possible
interpretations. First, one can interpret it as representing the
opportunity cost the borrower suffers from storing collateral in the
lending pool, rather than having it available now. It could also be
interpreted as a present-bias parameter, representing the fact that
the borrower simply prefers $x$ dollars now, even when compared to
money of the same value ($xe^{rt}$) in the future.  With a discount
parameter, topping-up is no longer utility-neutral in a
vacuum. Depositing some collateral now to get the same amount back
later yields negative utility, so agents would only want to top-up as
much as is necessary to avoid liquidation (and may prefer liquidation
to very high top-ups). 

Another way to fix the issue would be to consider a cost of financing,
rather than a discount parameter. More specifically, assume that agent
starts with no money, but can borrow money from an external lender
(i.e., not the lending pool), at a continuously compounding interest
rate of $\delta$. Now $\delta$ represents the cost of financing, both
for the original collateral and for any additional top-up
collateral. We choose to introduce a discount parameter rather than a
financing parameter because the former makes more sense when
replicating our loan as an option. The financing parameter involves a
separate entity (a traditional lender offering loans, for example),
and so the $\delta$ parameter is not as easy to relate to the main
parameter we care about: $\alpha$. 

\subsection{Discretization}
The second issue with the model is that when agents want to top-up, they would prefer to add infinitesimal amounts of collateral. To see why, suppose that the price of the asset reaches the liquidation threshold, and that the borrower (with some discount parameter $\delta$) wants to avoid liquidation. How much should they top-up? With continuous monitoring, they always want to top-up as little as possible. If they top-up an infinitesimal amount, and the price continues falling, they can simply top-up more. And further top-ups will be \textit{cheaper}, both because the price is falling, and because the top-ups will occur in the future, and thus be discounted by their discount parameter. Since there is no risk of liquidation, agents would always want to top-up as little as possible.

In reality, borrowers are worried about the risk of liquidation, and don't want to top-up tiny amounts repeatedly. To model these realities more accurately, we switch to discrete time monitoring. Now, borrowers can only interact with their lending pool loans at discrete time intervals. At these intervals, they can decide to repay the loan or top-up. If the liquidation threshold is breached in between these intervals, the loan gets liquidated. With this change, borrowers would want to top-up larger amounts to tide them over to the next monitoring interval.

\subsection{Replication}
To replicate the loan with top-ups, we first describe a novel option,
the \textit{top-up} option. This option has all the features of a
perpetual American down-and-out call option. In addition, the borrower
of the option can perform a \textit{top-up} at any time, which is when
they give the seller $d$ stock in order to update the barrier to
$B/(1+D)$, where $B$ is the original barrier, and $D$ is the sum of
all the stock the seller has received through top-ups. The seller can
keep this stock if the barrier is breached, but must return the stock
if the option is exercised.  

\begin{theorem}
    \label{top-up-replication}
    A top-up option (as described above) with exercise price $E =
    e^{\alpha t}S_0/c + \beta$, value $v = (1-1/c)S_0$, and barrier $B
    = (e^{\alpha t}S_0/c + \beta)c_0$ can be used to replicate both
    the borrower and lender's utility for a dynamic, perpetual lending
    pool with top-ups and a fixed fee, even if borrowers have a
    discount factor and monitor the loan discretely. 
\end{theorem}

Discretization and discounting are both trivial to translate to the option world; we simply assume that option participants discount by an additional factor of $\delta$, and that option buyers also monitor discretely, at the same rate as the borrowers. Neither of these assumptions changes the option parameters needed for replication, and the replication is exactly the same as \autoref{fixed_fee_repl}, so we omit the full proof. The main takeaway from this theorem is that we can use simulations to find the value of these top-up options, and then set that equal to the same value of $(1-1/c)S_0$.

\section{Finding fair interest rates via simulation}
In this section, we describe pricing simulations for our model of dynamic, perpetual loans with top-ups. As mentioned earlier, we replicate loans using options with value $v = (1 - 1/c)S_0$. We first describe the simulations in more detail. Then, we show how the value of the option depends on values of $r$, $\sigma$, $\delta$, and the monitor frequency. Finally, we compare our model's fair interest rates to Aave's interest rates through the last year, using historical data for the volatility and risk-free rate.

\subsection{Simulation details}
We first describe our simulation without top-ups. We simulate a large number of price paths, where each price path is a discrete approximation of the geometric Brownian motion of the stock price. We simulate over a large time period, as we cannot simulate a truly perpetual option; but due to discounting, and due to the fact that the exercise price and barrier rise over time, buyers do not want to hold the option for very long. The first challenge to simulation is determining how the buyers would behave in a given day (given that the option has not breached the barrier). Since geometric Brownian motion is a Markov process, we know that the buyers do not learn anything extra from the price history as compared to just the current price. If the exercise price and barrier were constant, they would base their decision to exercise the option solely on the current price. Call this their exercise threshold. 

However, the exercise price and barrier increase over time (proportional to $e^{\alpha t}$), so a fixed exercise threshold is not rational. We assume that the optimal exercise threshold takes the form $e^{\alpha t} S^*$, where $S^*$ is a constant, so the agent is more likely to exercise early, while the exercise price and barrier are lower. The first phase of our simulation involves approximating this $S^*$. We do this by linear search: we simply estimate the value extracted from the option for a range of different values of $S^*$, and choose the best one. We test this over 40,000 price paths to estimate $S^*$. The second phase of the simulation simply tests this threshold over new price paths (typically 200,000 price paths), and reports the average value. 

Simulating top-ups is the second major challenge, as they have two important parameters: when should agents top-up, and how much should they top-up. These parameters affect (and are affected by) the optimal exercise threshold. Here, we make two simplifying assumptions to help make the simulations tractable. First, we assume very simple top-up strategies for the agent. We typically assume that agents will top-up 0.1 more collateral when the price gets within 5\% of the barrier. Second, we scale down the exercise threshold by the amount $D$ of total collateral. That is, they exercise the option on day $t$ if $S_t > e^{\alpha t}S^* / D$.




\subsection{Option values as a function of input parameters}
Here, we measure the option buyer's value as a function of the volatility ($\sigma$), discount rate ($\delta$), risk-free rate ($r$), and monitor frequency. The buyer's value is used as a proxy for the interest rate.\footnote{This is because our simulation acts as an explicit function for the buyer's value, but only an implicit function for the fair interest rate. We explain this in more detail in the next subsection.} The base parameters for the simulation are $r = 0.03746$, $c = 1/0.805$, $c_0 = 1/0.83$, $S_0 = 100$, $\sigma = 0.46$, $\alpha = 0.0283$, $\delta = 0.005$, $\beta = 0.5$, and a monitor frequency of 10 times daily. These parameters represent real market conditions from February 2023, with lending-pool parameters ($\alpha$, $c$, and $c_0$) taken from Aave's February 2023 data. We simulate for $5$ years; due to our discount rate, the vast majority of options are exercised before this date. 

\begin{figure}
  \centering
  \subfloat{\includegraphics[width=0.48\textwidth]{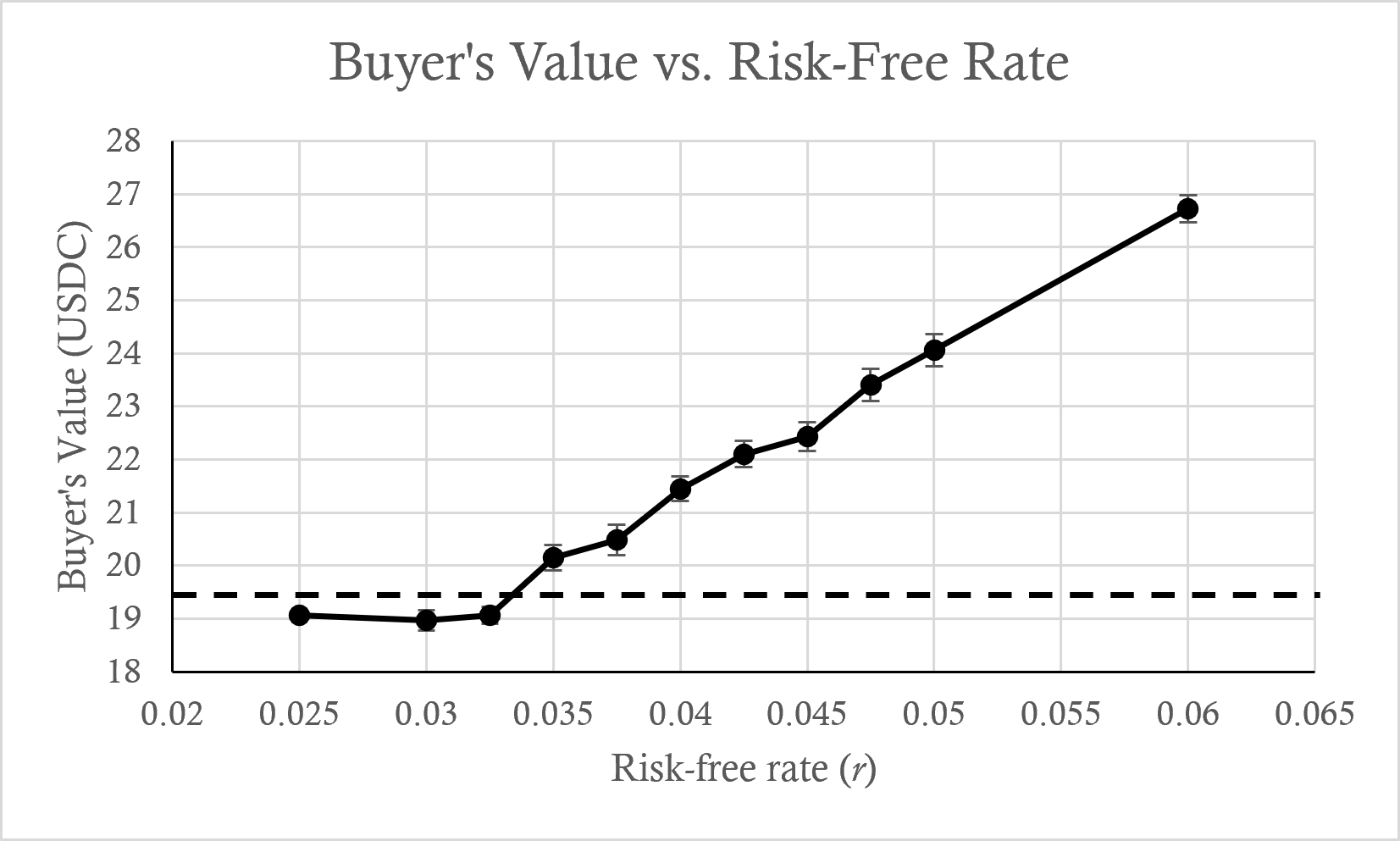}}
  \hfill
  \subfloat{\includegraphics[width=0.48\textwidth]{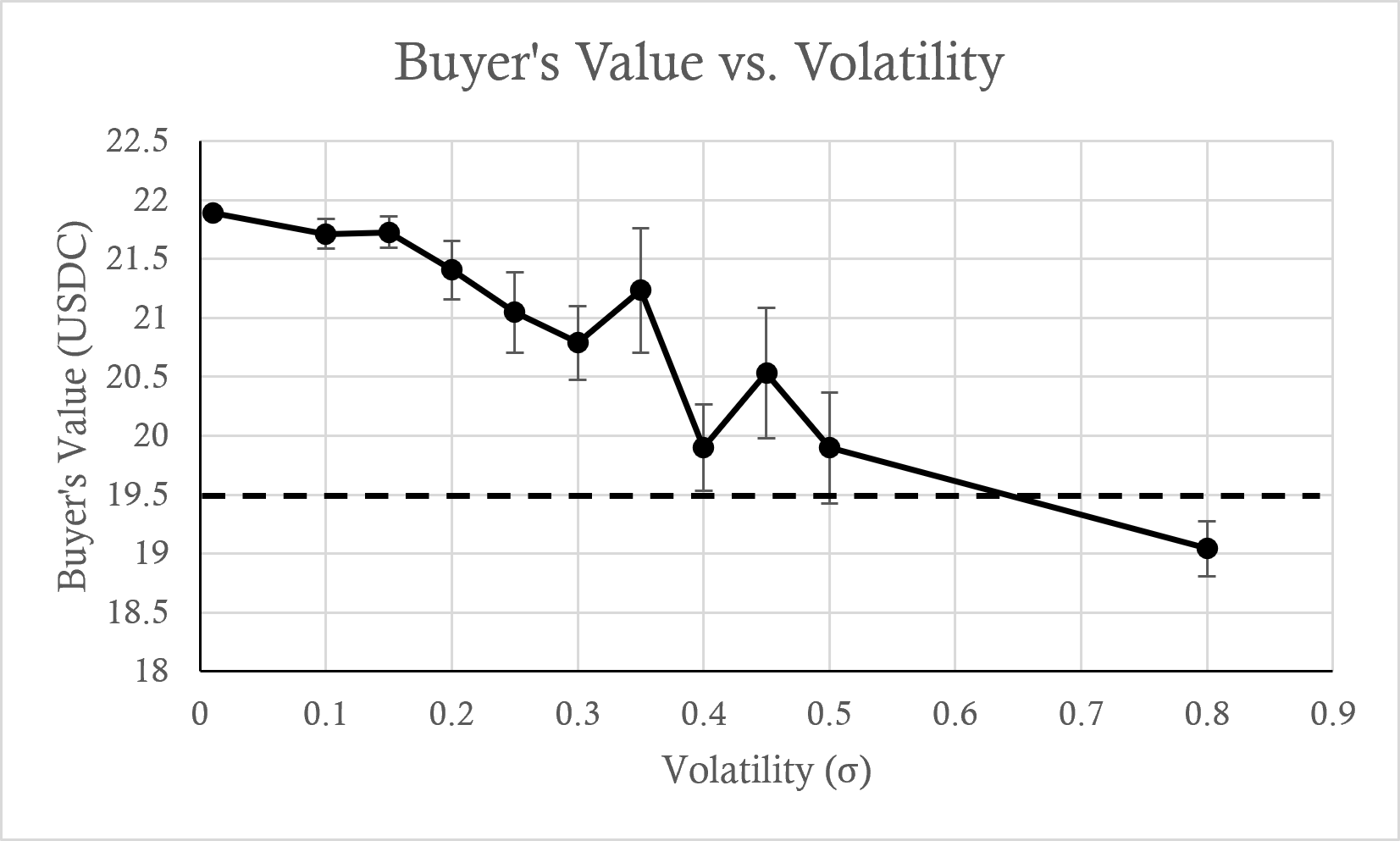}}
  \\
  \subfloat{\includegraphics[width=0.48\textwidth]{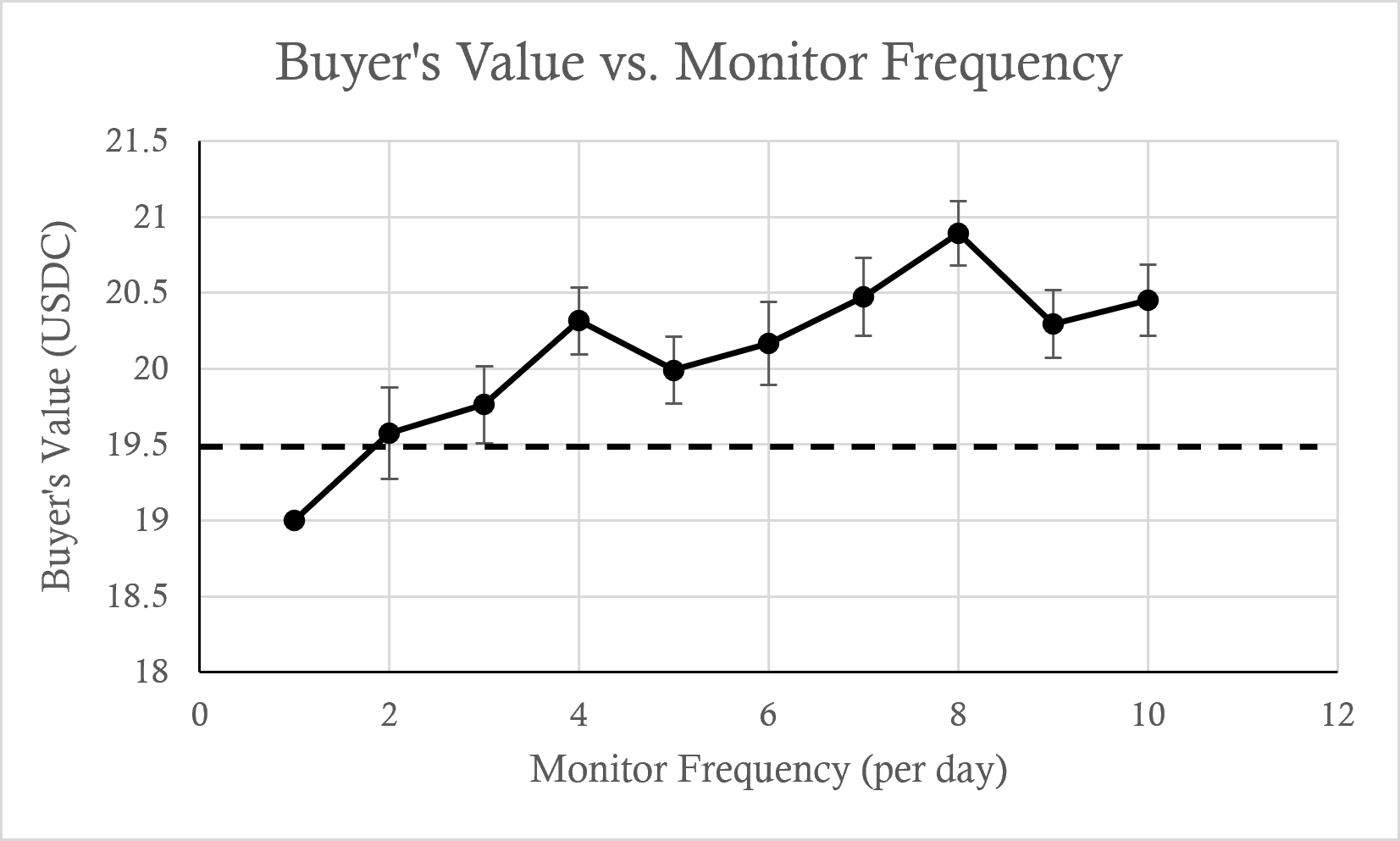}}
  \hfill
  \subfloat{\includegraphics[width=0.48\textwidth]{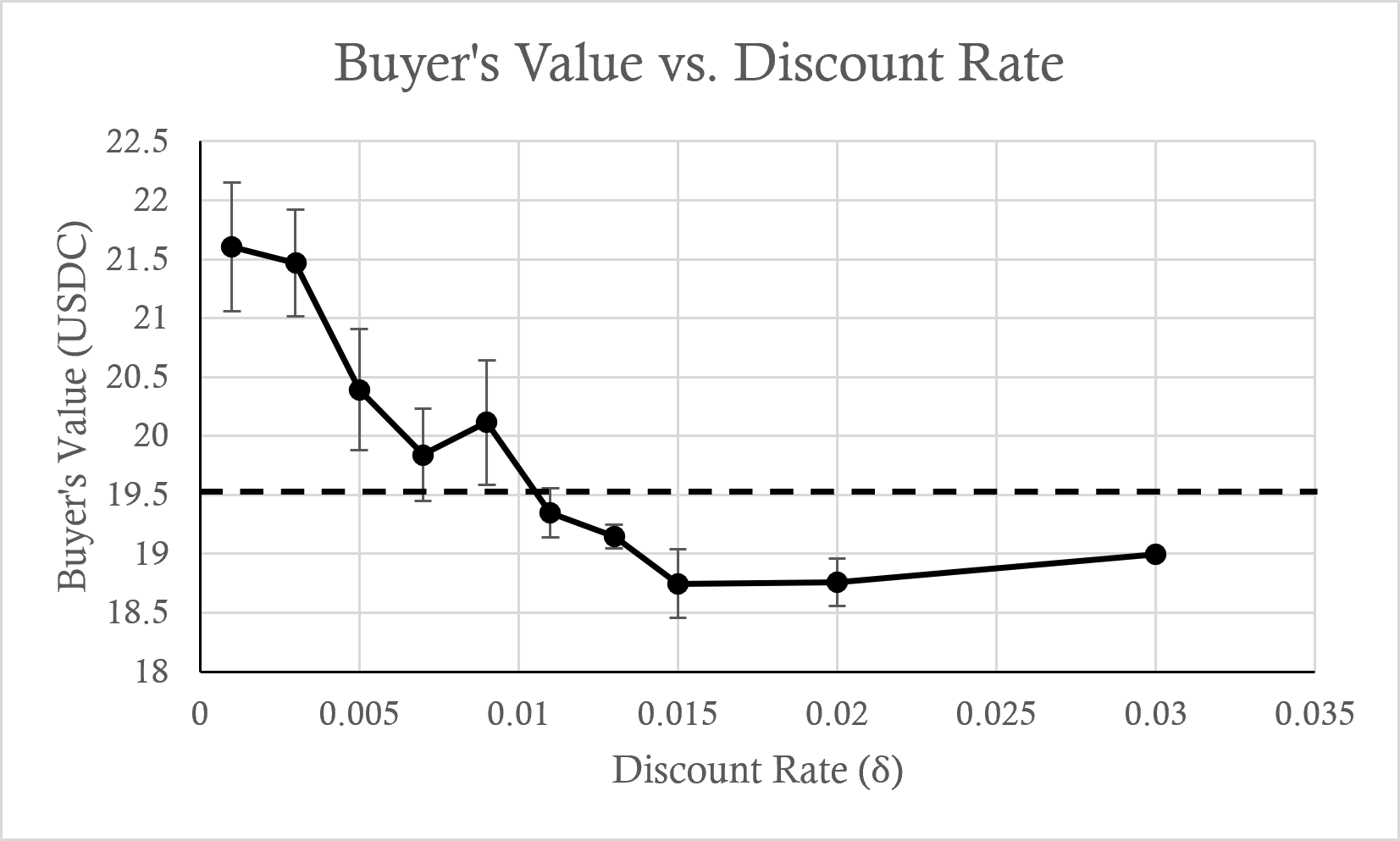}}
  \caption{These graphs show how buyer's expected value changes as different simulation parameters change. The base parameters for the simulation are $r = 0.03746$, $c = 1/0.805$, $c_0 = 1/0.83$, $S_0 = 100$, $\sigma = 0.46$, $\alpha = 0.0283$, $\delta = 0.005$, $\beta = 0.5$, and a monitor frequency of 10 times daily. The horizontal dashed line indicates the fair option value of $S_0(1-1/c) = 19.5$. The error bars represent the standard error of our test samples.}
  \label{fig:input_params}
\end{figure}

In \autoref{fig:input_params}, we plot our results. The dashed line indicates the option value of $S_0(1-1/c) = 19.5$ that is needed for replication. Thus, values above the dashed line indicate that a higher interest rate than $0.0283$ should be charged, while values below indicate that less interest should be charged (and values right on the line indicate fair lending-pool parameters). The top two graphs focus on market conditions, while the bottom two focus on buyer/borrower behavior. Note that, by our replication arguments, the buyer's utility equals the borrower's utility in the lending pool. We thus frame our discussion in terms of the borrower. 

We notice the following main trends. First, the rate of return $r$ is positively correlated with the borrower's value. Note that, in our model, $r$ equals the stock drift and informs how the borrower discounts future value (they discount at $r + \delta$). While these factors cancel out, the \textit{interest rate} stays constant in this graph. Thus, higher values of $r$ afford the borrower significantly higher utility. In fact, these graphs indicate that the borrower's value is most sensitive to $r$, out of all four parameters.

Second, borrowers prefer markets with lower volatility. This is fairly obvious, as lower volatility means that liquidations are less likely, and thus less money has to be spent on top-ups (which are utility-negative given a positive discount rate). However, there is not as much variation in the value for the range of realistic volatility parameters ($\sigma \in [0.2, 0.5]$ is representative of real market conditions from February 2023 to February 2024).

Third, relatively high monitor frequencies are necessary for borrowers. If borrowers monitored once a day, they would not even get positive value from the loan; as this increases to 8 times daily, they extract much higher value from the loan. The dip at the end could be simulation error, or it could represent how borrowers should edit their top-up strategies to be less conservative as they monitor more frequently. We did not do that; we used a simple top-up strategy of increasing the collateral by $0.1$ whenever the price was within $5\%$ of the liquidation threshold. 

Fourth, borrowers are also sensitive to the discount rate. As the discount rate increases to even moderate values like $0.015$, the borrower's utility quickly goes below the value needed for replication. As a result, we stick to a minor discount rate of $\delta = 0.005$. 

Overall, we model borrowers as having a discount rate of $0.005$ and monitoring 8 times daily. We also assume the pool has a fixed fee of $0.5$ \USDC. In the next subsection, we use these modeling choices to measure fair interest rates over the last year, using real market conditions. 

\subsection{Interest rates: model vs. real}
We now compare our model's proposed interest rates to Aave's real interest rates. We use Aave's \ETH\ liquidation parameter of $c_0 = 1/0.83$, their \ETH\ overcollateralization value of $c = 1/0.805$, and compare our model to their \USDC\ interest rates. We compute monthly averages for Aave's interest rates, getting the interest rate data directly from their app. For the market data, we use the 10-year U.S. treasury bond rate as the risk-free rate, and we estimate the annual \ETH\ volatility from 30-day averaged daily volatility.\footnote{Sources: risk-free rate: \url{https://www.macrotrends.net/2016/10-year-treasury-bond-rate-yield-chart}, Ethereum volatility: \url{https://buybitcoinworldwide.com/ethereum-volatility/}} Finally, we assume $\delta = 0.005$, $\beta = 0.5$, and that borrowers monitor 8 times daily. 

To compute the fair interest rate, given the other parameters in the model, we essentially perform a binary search over values of $\alpha$. It's easy to see that the option buyer's utility is decreasing in $\alpha$, so we simply perform our simulation for a candidate $\alpha$ and use the value output by the simulation to decide whether to increase or decrease $\alpha$. We repeat this until we are within $0.5\%$ of the fair option value, $S_0(1-1/c)$.\footnote{Note that the simulated buyer's value is proportional to $S_0$ as well, so $S_0$ is not an important parameter in our model.}

\begin{figure}
  \centering
  \includegraphics[width=0.85\textwidth]{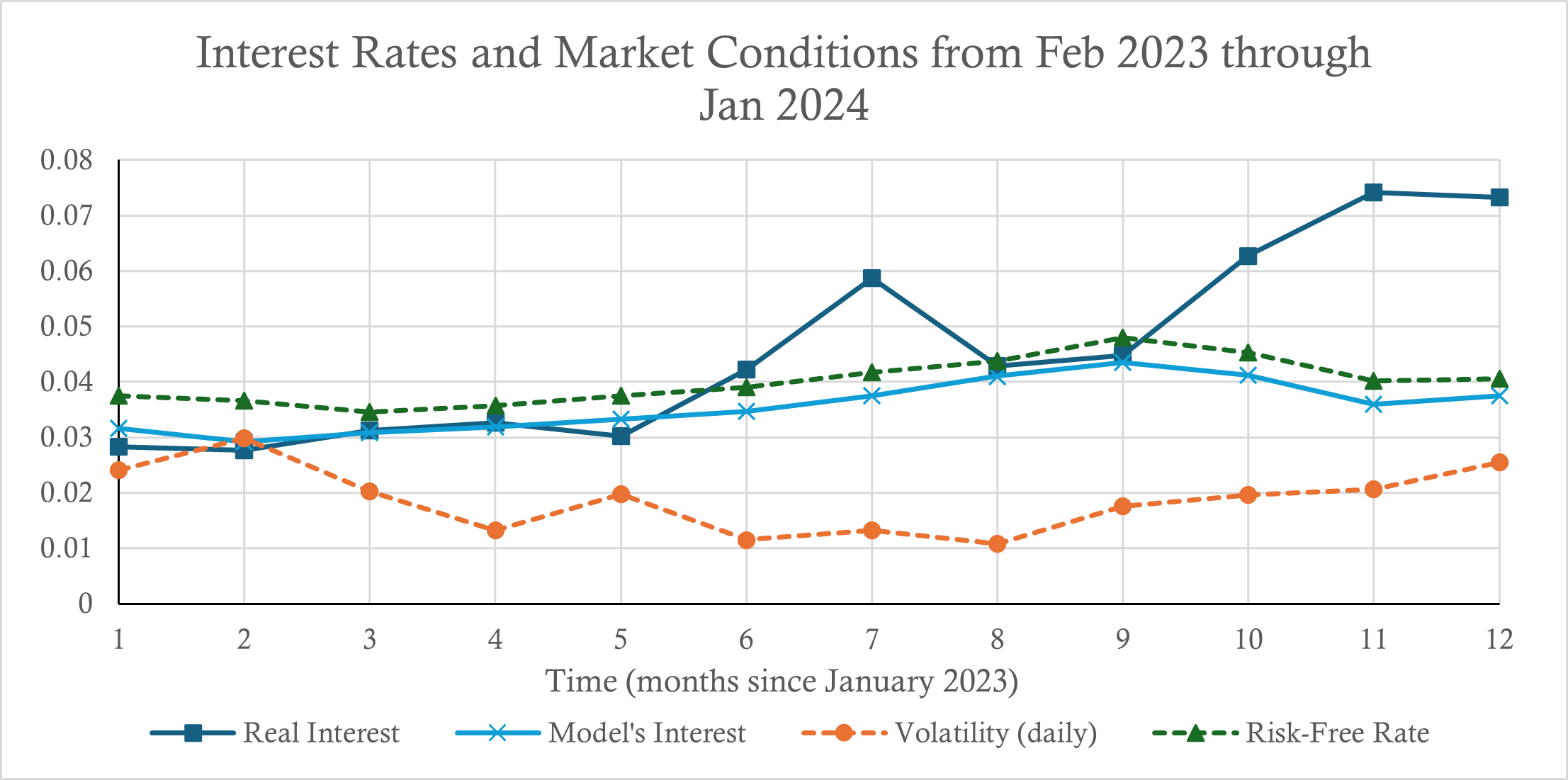}
  \caption{This graph compares our model's interest rates to the real interest rates in Aave's lending pool over the past year. The solid lines refer to the interest rates, and the dashed lines display the market conditions during each month.}
  \label{fig:final_model}
\end{figure}

\autoref{fig:final_model} summarizes our results. Our model's fair interest rates closely follow the risk-free rate, while the real interest rate experiences spikes in September 2023, December 2023, and January 2024. The Pearson correlation coefficient between the two models is $r = 0.575$ ($r^2 = 0.331$), with a p-value of $0.050$, indicating significant overlap between our interest rate model and practice. Thus, our model coheres closely with reality. While our model does predict lower interest rates on average, our model features a (modest) fixed fee. So, it's not clear whether the profits experienced by the lending pool would be lower -- that depends on the number of fees the lending pool collects, which depends on the duration of the average loan.

\begin{figure}
  \centering
  \subfloat{\includegraphics[width=0.48\textwidth]{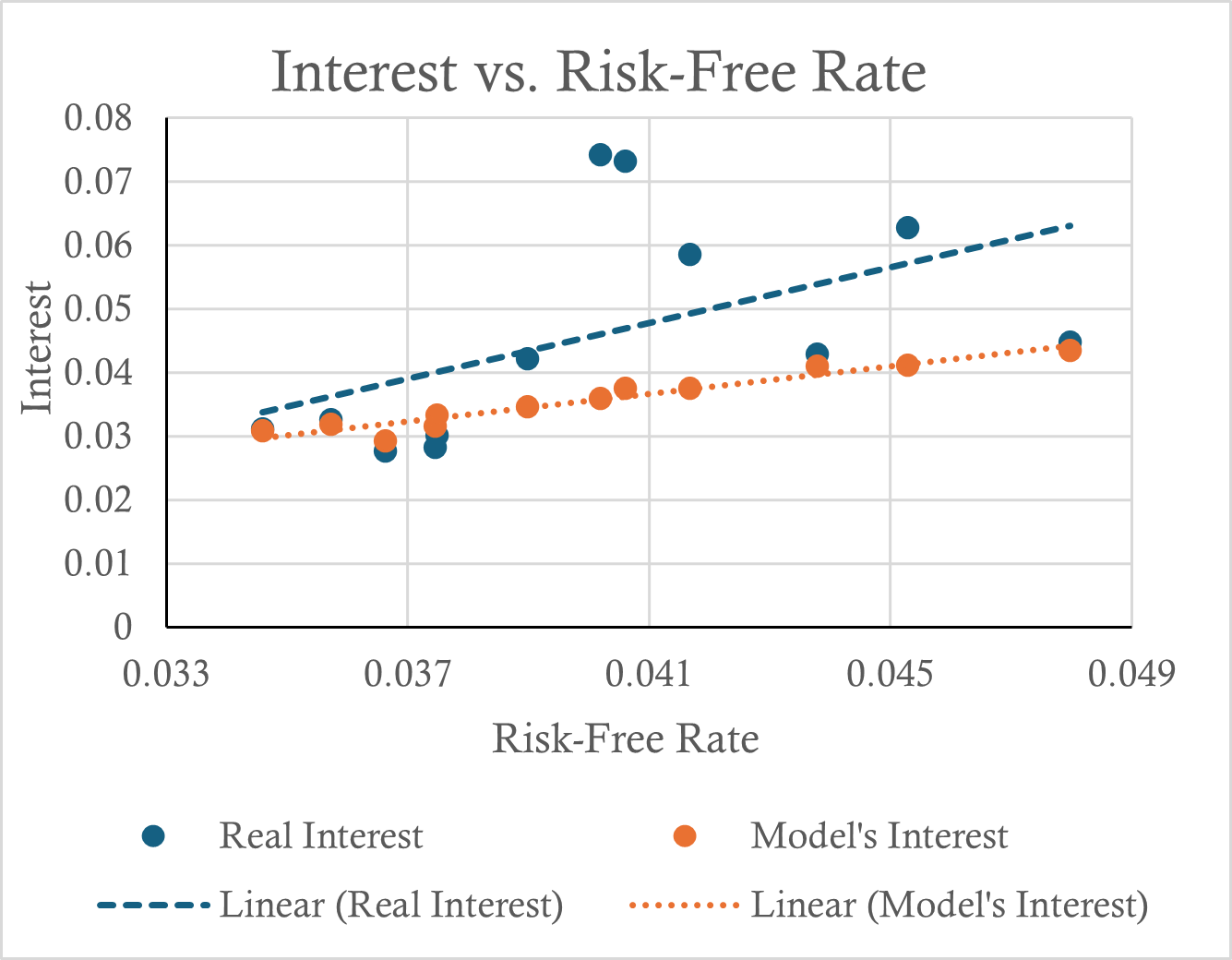}}
  \hfill
  \subfloat{\includegraphics[width=0.48\textwidth]{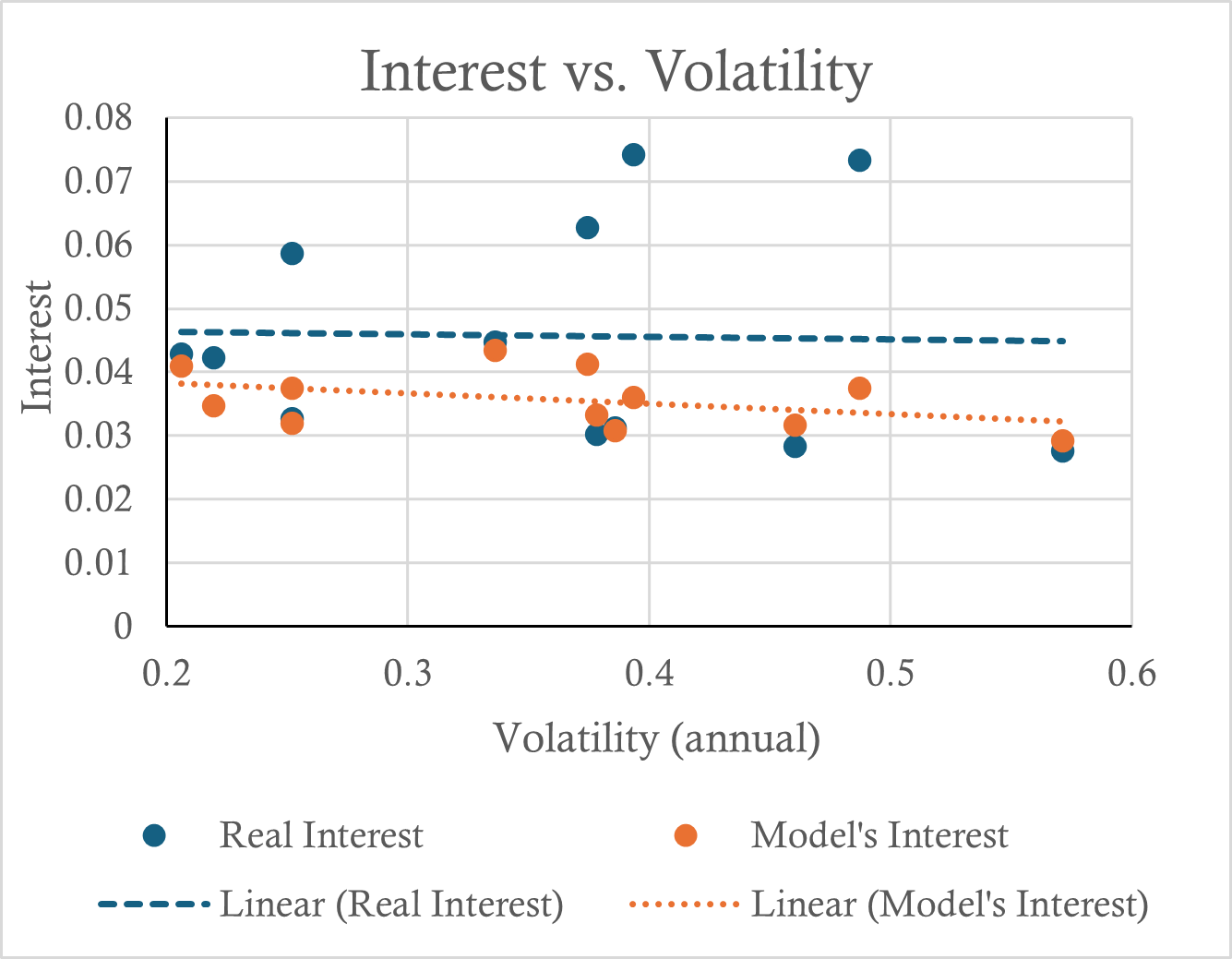}}
  \caption{This graph shows the relative impact of the risk-free rate and volatility. The dashed lines are the results of linear regressions.}
  \label{fig:final_model_parameters}
\end{figure}

\autoref{fig:final_model_parameters} plots the correlation between interest and the risk-free rate or volatility, for both our model and real interest rates. Our model is more tightly correlated with the risk-free rate ($r = 0.623, r^2 = 0.388732445, p = 0.030$) than the correlation between real interest rates and the risk-free rate ($r = 0.513, r^2 = 0.264, p = 0.088$). Further, real interest rates have little to no correlation to volatility (the $p$ value is $0.937$ on the $t$-test). Our model is at least weakly (negatively) correlated with volatility ($r = -0.401, r^2 = 0.161, p = 0.196$).

One might be tempted to conclude that real lending pools should also negatively correlate their interest rates with volatility. However, our model ignores the cost of liquidation for lenders, and in times of very high volatility, it's conceivable that lenders could lose significant assets due to market friction. So, the relationship between volatility and fair interest rates is more complicated in practice. Even if this cost to lenders is negligible in practice, the correlation of our model's interest rates with volatility is fairly weak. This could also help explain why existing lending pool are consistently popular, despite not lowering their interest rates in very volatile times. 

We emphasize that our results in this section should not be taken as concrete recommendations for lending pools. In practice, lending pools should at least estimate borrowers' average discount rates. Lending pools could also benefit from computing the trade-offs in profit for high fixed fees vs. high interest rates. The point of this section is to show how our model can serve as a practical starting point for more principled interest rate calculations. 


\section{Conclusion}
We investigated several models for finding fair lending pool interest rates using option replications. In a simplified model of fixed-term lending pools, we obtained analytic pricing results. But we found that in the more realistic model of dynamic lending pools, no fair parameters existed. We show that similar impossibility results hold in models where the borrowers get more or the lenders get less. To solve this issue, we added fixed fees to the dynamic lending pool model. We also added the ability for borrowers to top-up their loan, which required an extra discount parameter as well as a discrete monitoring paradigm. We then ran some simulations based on our final model, and compared our model's interest rates to real interest rates published by Aave over the last year, showing that our model is a good starting point for principled interest rate calculations. 

One advantage of using simulations to find fair interest rates is \textit{flexibility}. For instance, say that a subset of borrowers would be willing to pay more interest in exchange for lower overcollateralization rates (for example, those borrowers might want to more heavily increase their exposure to \ETH). Or some borrowers might want to pay less interest in exchange for a stricter liquidation threshold. Using simulations to estimate the fair interest rates allow for lending pools to offer \textit{dynamic} interest rates, which might vary depending on what lending-pool parameters the borrower wants. The Aave protocol already has separate borrowing modes depending on which currencies are being trading, so this approach would be a reasonable extension. 

We also identify several promising directions for further study. First, we used very simple option buyer strategies when building our simulations; the buyer's top-up strategy in particular was very simple. So one open question is to determine the optimal strategy for the option buyer in our model, given a set of model parameters.

Second, as we discussed in Section 4, our work focuses on identifying the \textit{fair} interest rate, which is the interest rate such that neither borrower nor lender can execute a series of trades with strictly higher utility. In practice, lenders may not be willing to trade their USDC for ETH. They may be risk-averse, or sensitive to transaction fees. For another example, financial institutions are often required to hold a certain amount of cash; perhaps stablecoins like \USDC\ may count towards this requirement in the future. In such settings, one might be more interested in \textit{revenue-maximizing} parameters, where one compares the revenue of loaning \USDC\ through the lending pool to simply holding \USDC\ and receiving the risk-free rate. Answering this question would also require one to look at the borrower's outside options more explicitly (to ensure their participation), a problem which we sidestep in the present work through option replication. 

Finally, even if one is interested in fair interest rates, there are other dynamic, perpetual loan models that don't involve adopting a fixed fee. For example, prospect theory is a popular behavioral economics theory that states that agents are risk-seeking when considering choices leading to losses but risk-averse when considering choices leading to gains \cite{kahneman1979prospect}. This can be applied to lending pools, since lenders experience more steady gains while borrowers are exposed to more loss. Thus, lenders may be more risk-averse while borrowers are risk-seeking. This could get around the impossibility result of \autoref{thm:perpetual_negative}. 

\bibliographystyle{ACM-Reference-Format}
\bibliography{ec2024/refs}


\begin{thebibliography}{14}


\ifx \showCODEN    \undefined \def \showCODEN     #1{\unskip}     \fi
\ifx \showDOI      \undefined \def \showDOI       #1{#1}\fi
\ifx \showISBNx    \undefined \def \showISBNx     #1{\unskip}     \fi
\ifx \showISBNxiii \undefined \def \showISBNxiii  #1{\unskip}     \fi
\ifx \showISSN     \undefined \def \showISSN      #1{\unskip}     \fi
\ifx \showLCCN     \undefined \def \showLCCN      #1{\unskip}     \fi
\ifx \shownote     \undefined \def \shownote      #1{#1}          \fi
\ifx \showarticletitle \undefined \def \showarticletitle #1{#1}   \fi
\ifx \showURL      \undefined \def \showURL       {\relax}        \fi
\providecommand\bibfield[2]{#2}
\providecommand\bibinfo[2]{#2}
\providecommand\natexlab[1]{#1}
\providecommand\showeprint[2][]{arXiv:#2}

\bibitem[\protect\citeauthoryear{Bastankhah, Nadkarni, Wang, Jin, Kulkarni, and
  Viswanath}{Bastankhah et~al\mbox{.}}{2024}]%
        {bastankhah2024thinking}
\bibfield{author}{\bibinfo{person}{Mahsa Bastankhah}, \bibinfo{person}{Viraj
  Nadkarni}, \bibinfo{person}{Xuechao Wang}, \bibinfo{person}{Chi Jin},
  \bibinfo{person}{Sanjeev Kulkarni}, {and} \bibinfo{person}{Pramod
  Viswanath}.} \bibinfo{year}{2024}\natexlab{}.
\newblock \showarticletitle{Thinking fast and slow: Data-driven adaptive defi
  borrow-lending protocol}.
\newblock \bibinfo{journal}{\emph{arXiv preprint arXiv:2407.10890}}
  (\bibinfo{year}{2024}).
\newblock


\bibitem[\protect\citeauthoryear{Black and Scholes}{Black and Scholes}{1973}]%
        {black-scholes}
\bibfield{author}{\bibinfo{person}{Fischer Black} {and} \bibinfo{person}{Myron
  Scholes}.} \bibinfo{year}{1973}\natexlab{}.
\newblock \showarticletitle{The pricing of options and corporate liabilities}.
\newblock \bibinfo{journal}{\emph{Journal of political economy}}
  \bibinfo{volume}{81}, \bibinfo{number}{3} (\bibinfo{year}{1973}),
  \bibinfo{pages}{637--654}.
\newblock


\bibitem[\protect\citeauthoryear{Chaudhary, Kozhan, and
  Viswanath-Natraj}{Chaudhary et~al\mbox{.}}{2023}]%
        {chaudhary2023interest}
\bibfield{author}{\bibinfo{person}{Amit Chaudhary}, \bibinfo{person}{Roman
  Kozhan}, {and} \bibinfo{person}{Ganesh Viswanath-Natraj}.}
  \bibinfo{year}{2023}\natexlab{}.
\newblock \showarticletitle{Interest rate rules in decentralized finance:
  Evidence from compound}.
\newblock \bibinfo{journal}{\emph{Open Access Series in Informatics (OASIcs)}}
  \bibinfo{volume}{110} (\bibinfo{year}{2023}), \bibinfo{pages}{5}.
\newblock


\bibitem[\protect\citeauthoryear{Chiu, Ozdenoren, Yuan, and Zhang}{Chiu
  et~al\mbox{.}}{2022}]%
        {chiu2022fragility}
\bibfield{author}{\bibinfo{person}{Jonathan Chiu}, \bibinfo{person}{Emre
  Ozdenoren}, \bibinfo{person}{Kathy Yuan}, {and} \bibinfo{person}{Shengxing
  Zhang}.} \bibinfo{year}{2022}\natexlab{}.
\newblock \showarticletitle{On the fragility of defi lending}.
\newblock \bibinfo{journal}{\emph{Available at SSRN 4328481}}
  (\bibinfo{year}{2022}).
\newblock


\bibitem[\protect\citeauthoryear{Cohen, S{\'a}nchez-Betancourt, and
  Szpruch}{Cohen et~al\mbox{.}}{2023}]%
        {cohen2023economics}
\bibfield{author}{\bibinfo{person}{Samuel~N Cohen}, \bibinfo{person}{Leandro
  S{\'a}nchez-Betancourt}, {and} \bibinfo{person}{Lukasz Szpruch}.}
  \bibinfo{year}{2023}\natexlab{}.
\newblock \showarticletitle{The economics of interest rate models in
  decentralised lending protocols}.
\newblock \bibinfo{journal}{\emph{Available at SSRN}} (\bibinfo{year}{2023}).
\newblock


\bibitem[\protect\citeauthoryear{Gao, Huang, and Subrahmanyam}{Gao
  et~al\mbox{.}}{2000}]%
        {gao2000valuation}
\bibfield{author}{\bibinfo{person}{Bin Gao}, \bibinfo{person}{Jing-zhi Huang},
  {and} \bibinfo{person}{Marti Subrahmanyam}.} \bibinfo{year}{2000}\natexlab{}.
\newblock \showarticletitle{The valuation of American barrier options using the
  decomposition technique}.
\newblock \bibinfo{journal}{\emph{Journal of Economic Dynamics and Control}}
  \bibinfo{volume}{24}, \bibinfo{number}{11-12} (\bibinfo{year}{2000}),
  \bibinfo{pages}{1783--1827}.
\newblock


\bibitem[\protect\citeauthoryear{Gudgeon, Werner, Perez, and
  Knottenbelt}{Gudgeon et~al\mbox{.}}{2020}]%
        {gudgeon2020defiprotocolsloanablefunds}
\bibfield{author}{\bibinfo{person}{Lewis Gudgeon}, \bibinfo{person}{Sam~M.
  Werner}, \bibinfo{person}{Daniel Perez}, {and} \bibinfo{person}{William~J.
  Knottenbelt}.} \bibinfo{year}{2020}\natexlab{}.
\newblock \bibinfo{title}{DeFi Protocols for Loanable Funds: Interest Rates,
  Liquidity and Market Efficiency}.
\newblock
\newblock
\showeprint[arxiv]{2006.13922}~[q-fin.GN]
\urldef\tempurl%
\url{https://arxiv.org/abs/2006.13922}
\showURL{%
\tempurl}


\bibitem[\protect\citeauthoryear{Haug}{Haug}{2007}]%
        {Haug_2007}
\bibfield{author}{\bibinfo{person}{Espen~Gaarder Haug}.}
  \bibinfo{year}{2007}\natexlab{}.
\newblock \bibinfo{booktitle}{\emph{The Complete Guide to option pricing
  formulas}}.
\newblock \bibinfo{publisher}{McGraw-Hill}.
\newblock


\bibitem[\protect\citeauthoryear{Kahneman and Tversky}{Kahneman and
  Tversky}{1979}]%
        {kahneman1979prospect}
\bibfield{author}{\bibinfo{person}{Daniel Kahneman} {and} \bibinfo{person}{Amos
  Tversky}.} \bibinfo{year}{1979}\natexlab{}.
\newblock \showarticletitle{Prospect Theory: An Analysis of Decision under
  Risk}.
\newblock \bibinfo{journal}{\emph{Econometrica}} \bibinfo{volume}{47},
  \bibinfo{number}{2} (\bibinfo{year}{1979}), \bibinfo{pages}{263--292}.
\newblock


\bibitem[\protect\citeauthoryear{Qin, Ernstberger, Zhou, Jovanovic, and
  Gervais}{Qin et~al\mbox{.}}{2023}]%
        {qin2023mitigating}
\bibfield{author}{\bibinfo{person}{Kaihua Qin}, \bibinfo{person}{Jens
  Ernstberger}, \bibinfo{person}{Liyi Zhou}, \bibinfo{person}{Philipp
  Jovanovic}, {and} \bibinfo{person}{Arthur Gervais}.}
  \bibinfo{year}{2023}\natexlab{}.
\newblock \showarticletitle{Mitigating decentralized finance liquidations with
  reversible call options}.
\newblock \bibinfo{journal}{\emph{arXiv preprint arXiv:2303.15162}}
  (\bibinfo{year}{2023}).
\newblock


\bibitem[\protect\citeauthoryear{Qin, Zhou, Gamito, Jovanovic, and Gervais}{Qin
  et~al\mbox{.}}{2021}]%
        {empirical_defi}
\bibfield{author}{\bibinfo{person}{Kaihua Qin}, \bibinfo{person}{Liyi Zhou},
  \bibinfo{person}{Pablo Gamito}, \bibinfo{person}{Philipp Jovanovic}, {and}
  \bibinfo{person}{Arthur Gervais}.} \bibinfo{year}{2021}\natexlab{}.
\newblock \showarticletitle{An Empirical Study of DeFi Liquidations:
  Incentives, Risks, and Instabilities}. In
  \bibinfo{booktitle}{\emph{Proceedings of the 21st ACM Internet Measurement
  Conference}} (Virtual Event) \emph{(\bibinfo{series}{IMC '21})}.
  \bibinfo{publisher}{Association for Computing Machinery},
  \bibinfo{address}{New York, NY, USA}, \bibinfo{pages}{336–350}.
\newblock
\showISBNx{9781450391290}
\urldef\tempurl%
\url{https://doi.org/10.1145/3487552.3487811}
\showDOI{\tempurl}


\bibitem[\protect\citeauthoryear{Rivera, Saleh, and Vandeweyer}{Rivera
  et~al\mbox{.}}{2023}]%
        {rivera2023equilibrium}
\bibfield{author}{\bibinfo{person}{Thomas~J Rivera}, \bibinfo{person}{Fahad
  Saleh}, {and} \bibinfo{person}{Quentin Vandeweyer}.}
  \bibinfo{year}{2023}\natexlab{}.
\newblock \showarticletitle{Equilibrium in a defi lending market}.
\newblock \bibinfo{journal}{\emph{Available at SSRN 4389890}}
  (\bibinfo{year}{2023}).
\newblock


\bibitem[\protect\citeauthoryear{Sardon}{Sardon}{2021}]%
        {sardon2021zero}
\bibfield{author}{\bibinfo{person}{Aetienne Sardon}.}
  \bibinfo{year}{2021}\natexlab{}.
\newblock \showarticletitle{Zero-Liquidation Loans: A Structured Product
  Approach to DeFi Lending}.
\newblock \bibinfo{journal}{\emph{arXiv preprint arXiv:2110.13533}}
  (\bibinfo{year}{2021}).
\newblock


\bibitem[\protect\citeauthoryear{Yaish, Dotan, Qin, Zohar, and Gervais}{Yaish
  et~al\mbox{.}}{2023}]%
        {subopt2023}
\bibfield{author}{\bibinfo{person}{Aviv Yaish}, \bibinfo{person}{Maya Dotan},
  \bibinfo{person}{Kaihua Qin}, \bibinfo{person}{Aviv Zohar}, {and}
  \bibinfo{person}{Arthur Gervais}.} \bibinfo{year}{2023}\natexlab{}.
\newblock \bibinfo{title}{Suboptimality in {DeFi}}.
\newblock \bibinfo{howpublished}{Cryptology {ePrint} Archive, Paper 2023/892}.
\newblock
\urldef\tempurl%
\url{https://eprint.iacr.org/2023/892}
\showURL{%
\tempurl}


\end{thebibliography}

\end{document}